\documentclass{article}

\usepackage{amsmath, amsthm, amscd, amsfonts, amssymb, graphicx, color}
\usepackage[bookmarksnumbered, colorlinks, plainpages]{hyperref}
\usepackage{pspicture,pst-all,multido}
\usepackage{slashed}

\usepackage{tikz}
\usetikzlibrary{calc}
\usepackage{float}

\definecolor{surfacecolor}{HTML}{222222}
\definecolor{deltacolor}{HTML}{2B9B52}
\definecolor{deltaprimecolor}{HTML}{D08A16}
\definecolor{causalcolor}{HTML}{B44747}
\definecolor{causallabel}{HTML}{9F3D3D}
\definecolor{backgroundcolor}{HTML}{FBFBFA}
\theoremstyle{definition}

\theoremstyle{remark}

\numberwithin{equation}{section}

\newcommand{\SMALL}{\scriptsize}

\newcommand{\C}{\mathbb{C}}
\newcommand{\R}{\mathbb{R}}
\newcommand{\N}{\mathbb{N}}

\newcommand{\norm}[1]{\parallel\!\!#1\!\!\parallel}
\renewcommand{\d}{\operatorname{d}}
\newcommand{\e}{\operatorname{e}}
\renewcommand{\i}{\operatorname{i}}

\newcounter{envcount}%

{\vspace{\bigskipamount}\refstepcounter{envcount}\textbf{(\theenvcount)\enspace Asymptotic causality.}}%
  {\vspace{\bigskipamount}}

{\vspace{\bigskipamount}\refstepcounter{envcount}\textbf{(\theenvcount)\enspace Alternative proof of local Lipschitzness  of $\sigma$ and $\mathfrak{h}$ and openness of the domain $U$.}}%
  {\vspace{\bigskipamount}}

\newenvironment{Def}%
{\vspace{\bigskipamount}\refstepcounter{envcount}\textbf{(\theenvcount)\enspace Definition.}}%
  {\vspace{\bigskipamount}}

{\vspace{\bigskipamount}\refstepcounter{envcount}\textbf{(\theenvcount)\enspace Example.}}%
  {\vspace{\bigskipamount}}

{\vspace{\bigskipamount}\refstepcounter{envcount}\textbf{(\theenvcount)\enspace Question.}}%
  {\vspace{\bigskipamount}}

{\vspace{\bigskipamount}\refstepcounter{envcount}\textbf{(\theenvcount)\enspace Conjecture.}}%
  {\vspace{\bigskipamount}}

{\vspace{\bigskipamount}\refstepcounter{envcount}\textbf{(\theenvcount)\enspace Remarks.}}%
  {\vspace{\bigskipamount}}

\newenvironment{Rem}%
{\vspace{\bigskipamount}\refstepcounter{envcount}\textbf{(\theenvcount)\enspace Remark.}}%
  {\vspace{\bigskipamount}}

{\vspace{\bigskipamount}\refstepcounter{envcount}\textbf{(\theenvcount)\enspace Estimations.}}%
  {\vspace{\bigskipamount}}

\newenvironment{P7T}%
{\vspace{\bigskipamount}\refstepcounter{envcount}\textbf{(\theenvcount)\enspace Proof of (\ref{MTPVAL}) Theorem.}}%
  {\vspace{\bigskipamount}}

  \newenvironment{P8T}%
{\vspace{\bigskipamount}\refstepcounter{envcount}\textbf{(\theenvcount)\enspace Proof of (\ref{PLFPOL})  Theorem.}}%
  {\vspace{\bigskipamount}}

{\vspace{\bigskipamount}\refstepcounter{envcount}\textbf{(\theenvcount)\enspace Proof of Estimate} }%
  {\vspace{\bigskipamount}}

\newenvironment{The}%
{\vspace{\bigskipamount}\refstepcounter{envcount}\textbf{(\theenvcount)\enspace Theorem.}\itshape}%
  {\vspace{\bigskipamount}\upshape}

{\vspace{\bigskipamount}\refstepcounter{envcount}\textbf{(\theenvcount)\enspace Theorem}\itshape}%
  {\vspace{\bigskipamount}\upshape}
  
\newenvironment{Pro}%
{\vspace{\bigskipamount}\refstepcounter{envcount}\textbf{(\theenvcount)\enspace Proposition.}\itshape}%
  {\vspace{\bigskipamount}\upshape}

\newenvironment{Cor}%
{\vspace{\bigskipamount}\refstepcounter{envcount}\textbf{(\theenvcount)\enspace Corollary.}\itshape}%
  {\vspace{\bigskipamount}\upshape}

\newenvironment{Lem}%
{\vspace{\bigskipamount}\refstepcounter{envcount}\textbf{(\theenvcount)\enspace Lemma.}\itshape}%
  {\vspace{\bigskipamount}\upshape}

{\vspace{\bigskipamount}\refstepcounter{envcount}\textbf{(\theenvcount)\enspace Suffix.}\itshape}%
  {\vspace{\bigskipamount}\upshape}

{\vspace{\bigskipamount}\refstepcounter{envcount}\textbf{(\theenvcount)\enspace Lorentz contraction.}\itshape}%
  {\vspace{\bigskipamount}\upshape}

{\vspace{\bigskipamount}\refstepcounter{envcount}\textbf{(\theenvcount)\enspace Example.}\itshape}%
  {\vspace{\bigskipamount}\upshape}

\theoremstyle{definition}
\swapnumbers

\begin{document}
\setcounter{page}{1}
\pagenumbering{arabic}

\vspace{3cm}

\begin{center}

{\Large  Representation of the causal logic  for the Dirac system and the electron}

\vspace{0.5cm}
Domenico P.L. Castrigiano$^{a1}$, Carmine De Rosa$^{b2}$ (*),  Valter Moretti$^{b3}$\

\vspace{0.2cm}

$^a$Technische Universit\"at M\"unchen, Fakult\"at f\"ur Mathematik, M\"unchen, Germany\\ 
$^b$Dipartimento di Matematica, Universit\`a di Trento and TIFPA-INFN, Trento, Italy\\
(*) = corresponding author
\smallskip

{\it E-mail addresses}: {\tt   $^1$castrig\,\textrm{@}\,ma.tum.de, \: $^2$carmine.derosa@unitn.it\: $^3$valter.moretti@unitn.it}
\end{center}

\begin{abstract}
We construct and analyze a covariant representation of the causal logic for the Dirac system and for the electron. The construction is based on the conserved Dirac probability-density current and on the general current-to-localization procedure for achronal regions. A polynomial decay estimate for the current, obtained by a non-stationary phase argument, verifies the hypotheses needed to define a covariant achronal localization for the full Dirac system. The corresponding representation of the causal logic is then obtained by passing from achronal localization to complete spacetime regions. Restricting this structure to the positive-energy invariant subspace yields the covariant achronal localization and the associated causal-logic representation for the electron.

Several structural properties of these localization observables are established. On Euclidean space the Dirac localization coincides with the canonical projection-valued localization, while its position operator differs from the Newton--Wigner operator by an explicitly determined bounded self-adjoint correction. The full Dirac achronal localization is shown to be projection valued and to represent achronal separateness by orthogonality, hence satisfying the corresponding microscopic causality condition. After compression to the electron subspace, the localization becomes positive-operator valued but still satisfies the causality condition. We also extend known localization properties from spatial regions to general achronal regions, prove separation results and norm-one criteria, and study the high-boost limit, obtaining Lorentz contraction in a precise probabilistic sense. Finally, we discuss the electron--positron decomposition of the Dirac system and the state transformation induced by position measurements.  The result on the positron production caused by a measurement of the position of the electron is considered to be universally valid as it does not depend on peculiarities of the measuring device.  
\end{abstract}

\section{Introduction}
The symmetry group of Minkowski spacetime is the Poincar\'e group\footnote{Dilations are ruled out by fixing the scales for length and time.}.  Their (continuous unitary Hilbert space) representations  determine the relativistic quantum mechanical systems and furnish explicit expressions for the kinematic observables energy, linear momentum, angular momentum and Lorentz booster. The crucial position observable is missing. This gives rise to  a long-standing
problem, the localization of relativistic quantum particles. Mainly there are the requirements of causality with which localization has to comply, and which we like to discuss in the following. Again Minkowski spacetime furnishes the structure which unambiguously states the localization that meets completely  the causality conditions, namely  a Poincar\'e covariant representation of the  so-called causal logic  as discussed below. 
\\
\hspace*{6mm} 
{\bf Newton-Wigner position.}  The {\em Newton--Wigner} construction provides (under some regularity hypotheses) a  distinguished
position observable for  massive and massless helicity zero elementary quantum relativistic systems \cite{NW49}.  Wightman's seminal formulation  of {\it particle  localization}  \cite{W62}   
furnishes  a rigorous mathematical frame for Newton-Wigner's construction.
 The corresponding Newton-Wigner  localization (NWL) assigns to every region (Borel set) $\Delta$ of Euclidean space  a {\bf localization operator} $T(\Delta)$ acting on the Hilbert space $\mathcal{H}$ of states. It is an orthogonal projector, whose  expectation value $\langle \psi, T(\Delta) \psi\rangle$
 is regarded as the probability of localization  in $\Delta$ of the system in the state represented by the unit vector  $\psi \in \mathcal{H}$. Accordingly the localization operators 
 form a Euclidean covariant orthogonal projector valued normalized measure (PVM) on Euclidean space $T$. 
 By  Poincar\'e covariance $T$ extends uniquely  to all flat spacetime regions, i.e. Borel subsets of  spacelike hyperplanes of Minkowski spacetime \cite{CM82}, \cite[sec.\,3.1]{C17}.
\\
\hspace*{6mm}
{\bf NWL's failure.} Unfortunately there is a serious causality  issue with NWL known from the start. NWL localizes {\it frame-dependently}  \cite{NW49}, \cite{WM63}.  This means that for every state of the system there is at most one spacelike hyperplane, where the system is localized in a bounded region. It implies 
as a special case  the phenomenon of {\it instantaneous spreading}  \cite{WS55},  \cite{WM63}. In fact a wave-function, which  vanishes outside a bounded region, any instant later extends up to infinity.
\\
\hspace*{6mm}
{\bf Hegerfeldt's insight.}
As worked out by Hegerfeldt \cite{HR80,H85,H01}, it is 
the positivity of energy  which brings about  that a relativistic system initially  localized in a bounded region would jump instantaneously everywhere.  This acausal behavior rules out any localization based on a Euclidean covariant PVM  $T$ as there is a plenty of states with the system  localized in a bounded region.
\\
\hspace*{6mm}
It should be stressed that Hegerfeldt-type arguments do not have a
single universal interpretation.  In particular, in the context of
{\em Fermi's two-atom problem}, Buchholz and Yngvason showed that the
apparent causality paradox should disappear in a proper local
quantum-field-theoretic description of the experiment \cite{BY94} based on local algebra formalization; see
also \cite{H94}.  This point is important for separating different
issues. 
\\
\hspace*{6mm}
 The present paper is not concerned with causality of local
field observables in a full quantum field theory, but with the causal
behaviour of localization observables for relativistic quantum
mechanical systems, mainly single particles.  It is in this latter sense that the Hegerfeldt
issue  is relevant here.
\\
\hspace*{6mm}
This distinction is also emphasized in the algebraic approach to local
quantum field theory.  As discussed by Schroer \cite{S11}, causal
localization of observables is encoded in local algebras and in
properties such as primitive causality, the time-slice property and,
in local form, the causal shadow property.  These properties do not
impose direct restrictions on individual localization operators of
particles; rather, they concern the way in which the algebra of
observables associated with a spacetime region is determined by its
causal completion.  In this framework the Born--Newton--Wigner
localization of particles and the algebraic causal localization of
observables have different meanings and behave antagonistically at
finite times, while they are expected to coexist harmoniously only in
the timelike asymptotic regime relevant for scattering theory.  The
present paper does not address this field-theoretic reconciliation as instead proposed by one of the authors  in \cite{M26a},\,\cite{M26b}.  It
is concerned instead with the quantum-mechanical side of the problem,
namely with causality properties of localization observables for
relativistic particle systems.
\\
\hspace*{6mm}
{\bf Positive operator valued localization.}
Due to Hegerfeldt's insight one has to renounce  the possibility of  states  with the system  localized in a bounded region. Instead one
 considers a {\it regional probability of location}. It comes down to assume the localization operators $T(\Delta)$ to be more generally nonnegative operators 
 with $0\le T(\Delta)\le I$. Hence $T$ is  a positive operator valued normalized measure (POVM). In the literature  this kind of localization is also called an {\it unsharp localization}, see \cite{BGL97} and a brief review in \cite[sec.\,6]{C17}.
 Again $T$ extends in a Poincar\'e covariant manner uniquely to all
 flat spacetime regions \cite[(9) Theorem]{C17}.  
 \\
\hspace*{6mm}
 By the way, contrary to the NWL,  localizations  for {\it massless} particles based on   POVM  do exist \cite{C81}.
 Of course, instantaneous spreading still occurs  in case of a POV-localization if $T(\Delta)$ has  the eigenvalue $1$  for a bounded region $\Delta$.  
  Otherwise the probability of localization in a bounded region is  less than $1$ and no measurement of position can change that. However, physically relevant is its upper bound, which coincides with $||T(\Delta)||$. If  and only if $||T(\Delta)||=1$  the system can be localized within $\Delta$ by a suitable preparation, not strictly but as accurately as desired. $T$ is said to be {\bf separated} if  $||T(\Delta)||=1$ for all open $\Delta$.
 \\
\hspace*{6mm} 
{\bf Causality.}  It imposes much more than no eigenvalue $1$ of $T(\Delta)$ for a bounded region $\Delta$.
 Einstein causality in the most direct interpretation states  quite generally that the probability of localization in {{\em any region of influence}} $\Delta'$ determined by the limiting velocity of light
  (sec.\,\ref{CCMC} and fig.\,\ref{fig:region-of-influence}) is not less than that in the region $\Delta$ of actual localization
\begin{equation*}
T(\Delta)\le T(\Delta')  \tag{CC}
\end{equation*}
see \cite[sec.\,4]{C24}. As an example let $\Delta$ be a flat spacelike region and $\sigma$ a spacelike hyperplane.  The region of influence  of $\Delta$ in $\sigma$, denoted by $\Delta_\sigma$, is the set of points 
 in $\sigma$, which are not spacelike separated  from 
 $\Delta$.  Then CC imposes $T(\Delta)\le T(\Delta_\sigma)$. For more details see \cite[sec.\,7]{C17}.
\\
\hspace*{6mm}
The causality condition CC guarantees frame-independent localization, cf.\,\cite[(15) and ff.]{C17}. CC implies the  familiar {\it causal time evolution} by the special case that $\Delta$ and $\sigma$ are parallel. Hence
obviously it excludes instantaneous spreading. For more details see sec.\,\ref{CCMC}.
   POVM on flat spacelike regions which  satisfy CC are already known for the Dirac and Weyl fermions \cite{C17} and for the massive scalar boson \cite{C24}. The former localizations are separated \cite[sec.\,18.2, (100)]{C17}. Regarding the massive scalar boson this is not known \cite[sec.\,18]{C23}.
\\
\hspace*{6mm}
{\bf Achronal Localization} (AL). So far only flat spacelike regions are considered for localization. However this restriction is not physically justified. In \cite[sec.\,3]{C24},  on plain physical grounds  it is argued  in great detail  that localization has to concern all achronal regions of spacetime. Particularly the existence of the high boost limit, as a consequence of CC, induces the localization in  {\em not spacelike}  achronal hyperplanes. See \cite[sec.\,26]{C17} regarding the Dirac and Weyl fermions and \cite{C25} dedicated to the massive scalar boson.  Here in  sec.\,\ref{HBLALLA} a rather easy approach to the high boost limit is possible. Its existence implies the phenomenon of Lorentz contraction.
\\
\hspace*{6mm} 
This basic insight on achronal regions is decisive for the further development. An AL is a POVM on every maximal achronal set.  Most important,  AL automatically  satisfy CC \cite[(16) Theorem]{C24}. This is a consequence of the normalization of AL. See sec.\,\ref{NN} and for more details \cite[sec.\,2]{C24}. 
\\
\hspace*{6mm} 
{\bf Representation of the causal logic} (RCL). Minkowski spacetime is provided with the causal logic $\mathcal{C}$  introduced by Cegla and Jadczyk \cite{CJ77},  which  is the lattice  of  spacetime Borel sets generated by an orthogonality relation, the {\it achronal separateness}.  It is the set of all complete Borel sets. For details see sec.\,\ref{NN}. This structure is regarded to constitute the frame which complies most completely with  the principle of causality for massive quantum mechanical systems. Heuristically, a massive particle is imagined to be represented by a timelike straight world line. It  is considered to be localized in a spacetime region in $\mathcal{C}$ if its world line crosses the region. According to this simple heuristic picture, an RCL, which is a normalized $\sigma$-orthoadditive POVM  on $\mathcal{C}$, is a localization that meets all the causality conditions.
\\
\hspace*{6mm} 
In \cite{C24} it is shown that (covariant) AL and (covariant) RCL are in one-to-one
correspondence.  In concrete terms, if \(T\) is an AL and \(F\) is the
corresponding RCL, then
\[
        F(\Delta^\land)=T(\Delta)
\]
for every achronal Borel set \(\Delta\).  Here $\Delta^\land$ denotes the completion of $\Delta$ with respect to achronal separateness. Consequently, the
construction of a (covariant) AL is simultaneously the construction of a
(covariant) RCL.  
\\
\hspace*{6mm}
Let us point out that, in determining $F$ out of $T$, the fact is decisive that the completion of an achronal set coincides with its set of determinacy (\ref{AADAC}). 
\\
\hspace*{6mm}
For details on AL and RCL  see sec.\,\ref{NN} and \cite{C24}.
\\
\hspace*{6mm}
 {\bf Causal current.} The localization  of the massive scalar boson is ultimately extended to an AL  in \cite{CDM25}. Here we will achieve the AL of the (Dirac) electron (\ref{EMTDC}).
 Moreover, covariant AL are constructed for the systems with mass spectrum of positive Lebesgue measure and  every definite spin \cite[(27) Theorem]{C24}, which plainly are not elementary and are not  further considered here.
\\
\hspace*{6mm} 
There are two decisive works which enable the construction of (convariant) AL, \cite{GGP67} and \cite{DM24}. Petzold et al. recognized that the zeroth component of the standard real bilinear covariant conserved current for the massive scalar boson is indefinite and hence inadequate to represent a charged densitiy. So in \cite{GGP67} they reveal a basic one-parameter family of such currents with positive zeroth component called {\bf causal currents} for short. A causal current  describes a conserved probability density of localization.
For more details see \cite{C25}. The particular  merit of \cite{GGP67}  is to move  to the center the {\it fundamental concept} of a covariant conserved probability density current in place of a system of imprimitivity and its traces (compressions) on representation invariant subspaces. 
\\
\hspace*{6mm} 
 It was showed in \cite{DM24}  how to compute the flux of  a $C^1$\,causal current  passing through a $C^1$ surface and construct  the related  localization operator. (Actually that work assumes smoothness, but trivially $C^1$ is sufficient.)  In \cite{DM24} they obtain out of a causal current  a {\it coherent} 
POV-localization for the massive scalar boson  on spacelike $C^1$ Cauchy surfaces, which satisfy CC as a consequence of the normalization. Note that a Cauchy surface, which is a spacelike set, is maximal acronal. 
\\
\hspace*{6mm} 
In addition  a covariant family of spatial POV-localizations satisfying CC  for massive real scalar
Klein--Gordon particles, presented  in \cite{M23} out  of the stress energy tensor and improving a previous idea by Terno \cite{T14}, are extended to a covariant family of coherent POV-localizations  on spacelike $C^1$ Cauchy surfaces.
\\
\hspace*{6mm} 
{The results in \cite{DM24} were achieved by a suitable use of the divergence theorem for volume forms. It was done by  taking advantage of relatively recent advanced  results in Lorentzian geometry  concerning the extension of acausal manifolds with boundary to  spacelike smooth Cauchy surfaces.
\\
\hspace*{6mm} 
Now the construction of an achronal localization goes  still along the lines of the conserved-current  construction of the Cauchy localization in \cite{DM24}. It takes also  advantage of some sophisticated results of geometric measure theory.  In the last work  \cite{CDM25} all these ideas and mathematical techniques converge and a rather general result on the construction of covariant AL is achieved. 
\\
\hspace*{6mm}
 It is worth noting that a causal current generates quite naturally a nonnegative operator valued AL rather than a  projector-valued one. This matches perfectly the meaning of a POV-localization as a description of a regional probability  of location. Following this point of view the question does not arise whether the localization operators are observables, but the  question is whether  the flux of a causal current is observable.
 
  \subsection*{\bf Achievements  of this work}
$\bullet$\; The main objective, a covariant RCL for the electron is reached, see sec.\,\ref{CRCLDSE}. To this the corresponding AL  for the electron is determined applying  \cite[(19)]{CDM25}  to the well-known conserved Dirac probability density current. The polynomial decay of the current assumed in \cite[(19)]{CDM25} is shown to hold in sec.\,\ref{IMRC} by the non-stationary phase method. 
\\
$\bullet$\;  Actually one obtains in sec.\,\ref{CALDSE} a covariant  AL $T^\textsc{d}$ and a covariant RCL $F^\textsc{d}$ of the Dirac system $W^\textsc{d}$. It is the trace  $T^e$ of $T^\textsc{d}$ on the $W^\textsc{d}$-invariant subspace  of positive energy, which is the covariant AL of the electron. Analogously one obtains the covariant RCL $F^e$ of the electron.
\\
\hspace*{6mm}
 The following results concern   several  properties of these localizations treated in sec.\,\ref{PALDS}.
 \\
 $\bullet$\;  On the Euclidean space the localization
$T^\textsc d$ coincides with the canonical  PV-localization.  The
corresponding position operator $X^\textsc d$ is the multiplication operator with the coordinate vector $x\in\R^3$. In momentum 
representation it equals $\i\nabla$. In sec.\,\ref{DNWPO} one finds that the Newton--Wigner position operator $X^\textsc{nw}$, distinguished by commuting with the sign of energy,  differs from $X^\textsc d$
by a bounded selfadjoint operator $\i F$
\[
        X^\textsc{nw}=X^\textsc d+iF 
\]
In momentum representation $\i F$ is the multiplication operator by an explicitly given selfadjoint matrix function. A rough
analysis suggests that  the expectation values of $X^\textsc d$ and $X^\textsc{nw}$  with respect to electron states presumably are  experimentally
indiscernible.
\\
$\bullet$\;  In sec.\,\ref{POVDL} the achronal localization $T^\textsc d$ of the Dirac system is shown to be
projection valued.  This is obtained from a general result: if an AL is
projection valued on flat spacelike Borel sets and vanishes on sets
whose spatial projection has zero Lebesgue measure, then it is
projection valued. 
\\
$\bullet$\;
Achronal separateness and, a fortiori, spacelike separateness   is represented by orthogonality of the
localization projectors.  As noted  in sec.\,\ref{CCMC} the AL $T^\textsc{d}$ and the associated RCL $F^\textsc d$ therefore
satisfies local orthogonality and, in particular, the corresponding
microscopic causality.  
\\
\hspace*{6mm}
 These properties get lost at the compression of $T^\textsc{d}$ on the electron subspace. In
particular $ T^e$ of  is no longer projector valued. Of course,  the AL  $T^e$ still satisfies  CC.
\\
$\bullet$\;  The AL $T^e$ of the electron
is well studied  on Euclidean space, whereas 
$T^e$ on  achronal not flat regions is  uninvestigated.  A  spacetime region $\Delta$ is said to be {\em determining} if  there is  $\mathfrak{x} \not\in \Delta$ such that every timelike  straight line through $\mathfrak{x}$  meets $\Delta$. As shown a spacelike region is determining if and only if it has an interior point.
In sec.\,\ref{LPTE} two main results regarding flat spacelike regions 
are generalized to general achronal regions.
\\
\hspace*{6mm}
First, no electron state is
localized in a proper closed spatial region.  The generalization reads: No electron is localized in an achronal region for which there is an achronally separated  one, which is determining. In particular, non electron is localized in a bounded achronal region.
\\
\hspace*{6mm}
Second,  if a spatial 
region  has nonempty interior, then the  norm of the
localization operator is one.  The generalization reads: If an achronal region is
determining, then the norm of the
localization operator is one. In particular this holds true for a spacelike region with an interior point.
\\
\hspace*{6mm}
These result are rather satisfactory from a physical point of view as $||T^e(\Delta)||=1$ implies that there are  electron 
states with the localization probability in $\Delta$ arbitrarily close to one,  even if 
a localization in $\Delta$ is not possible, as e.g. for bounded $\Delta$.
\\
\hspace*{6mm}
Related to these results we just mention that $T^e$ is shown to be separated on every {\em causal base} -- namely a spacelike Cauchy surface 
(see  sec.\,\ref{NN}). --   and that for every point of the base an explicit 
point localized sequence of states is available.
\\
\hspace*{6mm}
Finally the determining subsets of  a light hyperplane are characterized as those containing an open achronal paraboloid.
\\
$\bullet$\; We study in sec.\,\ref{HBLALLA} the high boost limit and Lorentz contraction.  Spacelike
hyperplanes tend, under boosts of increasing rapidity, to light
hyperplanes.  We prove the corresponding convergence of localization
probabilities for the Dirac AL, in the spirit of the scalar-boson
analysis on achronal hyperplanes in \cite{C25}.  As a consequence,
both the full Dirac system and the electron exhibit Lorentz contraction
in the precise probabilistic sense that, for sufficiently large
boosts, the probability of finding the particle in an arbitrarily
narrow strip orthogonal to the boost direction tends to one. This phenomenon is discussed along the usual questions regarding classical Lorentz contraction.
\\
$\bullet$\; Finally we elaborate in sec.\,\ref{EPR} the idea that the Dirac system is constituted by the electron and the  positron. The positron subrep of the Dirac system is an electron rep transformed by the antiunitary time inversion operator. Hence the energy is positive. A state of the Dirac system, which is a nontrivial  linear combination of an electron and positron state, is regarded to be  a {\it non-observable virtual superposition} of these states. The attempt to realize this state by preparation creates the  respective {\em mixed state}. This is the mechanism by which a proton is created at the attempt to localize an electron in a bounded spacetime region.
\\
\hspace*{6mm}
The probabilities for the existence of an electron and  of a positron due to a position measurement are computed, and the density operator after the measurement  is determined. The results concern a {\em non-selective measurement}. They are {\em universally valid} in the sense that they hold true as long as no information gained from the actual measuring equipment is taken into account.
\\
$\bullet$\; The appendix contains the
technical results on unions of Lipschitz graphs, which are interesting in their own right,
pyramidal
approximations and the general criterion ensuring projection
operator-valuedness.

\section{Notations and notions}\label{NN}

{\bf Schwartz spaces}.
The {\bf support} $\overline{\{f\ne 0\}}$ of a continuous  function  $f$ is  denoted by  $\operatorname{supp}(f)$.
The symbol $\mathcal{S}(\R^n,\C^m)$ indicates the $\C^n$-{\bf Schwartz space}, i.e.   the space of $\C^n$-valued functions of rapid decrease on $\R^n$ and $C_c^\infty(\R^n,\C^m)$ the subspace of smooth function with compact support. For $m=1$ one simply writes $\mathcal{S}(\R^n)$, $C_c^\infty(\R^n)$. 
\\

{\bf Fourier transformation.}  It is given   by the formula $$(\mathcal{F}\psi)(p)=(2\pi)^{-n/2}\int_{\R^n}\e^{-i\,px}\psi(x)\d^nx\quad \mbox{where $p\in\R^n$}$$ if  $\psi: \R^n\to \C$ is (Lebesgue) integrable. Here and elsewhere $a\,b$ denotes the scalar product of $a,b\in\R^n$.
\\

 {\bf Lebesgue measure}. Occasionally the Lebesgue measure of the Borel subsets of $\R^n$ is denoted by $\mathcal{L}^n$.
\\

{\bf Minkowski spacetime}.
Minkowski spacetime is represented by $\R^4$ using the notation   $$\mathfrak{x}=(x_0,x)=(x_0,x_1,x_2,x_3)\in\R^4$$ 
Points
 in Minkowski spacetime are said {\bf events}.
When viewing the points in Minkowski spacetime as vectors defining translations, the {\bf Minkowski product} reads 
$$\mathfrak{x}\cdot \mathfrak{y}:= x_0y_0-xy= x_0y_0-x_1y_1-x_2y_2-x_3y_3\:, \quad \mbox{for $\mathfrak{x}, \mathfrak{y}\in \R^4$ }$$ 
Put $\mathfrak{x}^{\cdot 2} :=\mathfrak{x}\cdot \mathfrak{x}$. The {\bf canonical spatial projection} is $\varpi:\R^4\to \R^3$, $\varpi(\mathfrak{x}):=x$. 
 \\
 
{\bf Poincar\'e group.} $\tilde{\mathcal{P}}:=ISL(2,\C)$ denotes the {\em universal covering} of the {\em proper orthochronous Poincar\'e group} $\mathcal{P}^\uparrow_+$. 
$\tilde{\mathcal{P}}$ acts on
$\R^4$ 
 by $g\cdot \mathfrak{x}:= \mathfrak{a}+ A\cdot \mathfrak{x}=\mathfrak{a}+ \Lambda(A)\mathfrak{x}$, where $g=(\mathfrak{a},A)$ and $\Lambda$ is the covering homomorphism from  $SL(2,\C)$ onto 
the {\em proper orthochronous Lorentz group}   $\mathcal{L}^\uparrow_+$. In particular $A_\rho:=\e^{\,\rho\,\sigma_3/2}=\operatorname{diag}(\e^{\rho/2},\e^{-\rho/2})$ acts on $\R^4$ by $$\Lambda\big( A_\rho\big)=\left(\begin{array}{cccc}\cosh(\rho) & 0&0&\sinh(\rho)\\ 0&1&0&0\\0&0&1&0\\ \sinh(\rho)&0&0&\cosh(\rho)\end{array}\right)\quad  \mbox{for $\rho \in \R$.} $$
 $\tilde{\mathcal{E}}:=ISU(2)$ denotes the universal covering of the {\it Euclidean group}.
A {\bf representation} (rep) $W$ of  $\tilde{\mathcal{P}}$ and generally of a topological group is assumed to be unitary and strongly continuous. 
\\

 {\bf Achronal sets.} Two events  $\mathfrak{x},\mathfrak{y}\in\R^4$, $\mathfrak{x}\ne\mathfrak{y}$ are {\bf achronally separated} if $|x_0-y_0|\le |x-y|$. Similarly  sets $\Delta,\Gamma\subset\R^4$  are {\bf achronally separated} if 
(a) they are disjoint and (b)  $|x_0-y_0|\le |x-y|$ for  $\mathfrak{x}\in\Delta$, $\mathfrak{y}\in\Gamma$.   A set $\Delta\subset\R^4$ is called {\bf achronal} if 
$|x_0-y_0|\le |x-y|$ for $\mathfrak{x},\mathfrak{y}\in\Delta$. More specifically  it is called {\bf $L$-achronal} with $0\le L\le 1$ if $|x_0-y_0|\le L |x-y|$ for $\mathfrak{x},\mathfrak{y}\in\Delta$.
Obviously $\Delta$ is $L$-acronal if and only if it is the {\it graph} of an $L$-Lipschitz function.  
Note that every ($L$-)achronal set is contained in a {\bf maximal ($L$)-achronal} set, i.e., an ($L$)-achronal set being not a proper subset of any ($L$)-achronal set.
Note that a maximal ($L$)-achronal set is closed. Let $\mathcal{B}^{ach}$ denote the set of all Borel achronal subsets of $\R^4$.
For achronal $\Delta$
$$ \varpi|_\Delta: \Delta \to \varpi(\Delta)$$
is a homeomorphism. 
Clearly   an achronal  set $\Lambda$  is maximal if and only if $\varpi(\Lambda)=\R^3$. 
\\

{\bf Spacelike sets.}  Two distinct events  $\mathfrak{x},\mathfrak{y}\in\R^4$, $\mathfrak{x}\ne\mathfrak{y}$ are {\bf spacelike separated} if $|x_0-y_0| < |x-y|$.
  A straight line in $\R^4$ is  {\bf causal} if distinct points on it are not spacelike separated.
   A set $\Delta\subset\R^4$  is called {\bf spacelike} if $|x_0-y_0|<|x-y|$ for $\mathfrak{x},\mathfrak{y}\in\Delta$, $\mathfrak{x}\ne\mathfrak{y}$. Clearly every spacelike set is achronal.
       \\
\hspace*{6mm}    
 A spacelike set $\Sigma$ is called a  \textbf{causal base} if $\Sigma$ intersects all {\it causal}   straight lines.  By a result of Moretti in \cite[(9)]{C24}  the causal bases are exactly the {\em Cauchy surfaces} of Minkowski spacetime  which are a spacelike set.  
 \\
\hspace*{6mm}      
    Causal bases are maximal achronal. The set of causal bases is Poincar\'e invariant. The familiar spacelike hyperplanes are causal bases. Roughly speaking,  the causal bases are  those sets of independent events,   which  determine all events of Minkowski spacetime. For more details see \cite[sec.\,2.2.2]{C24}.
\\

{\bf Achronal localization} (AL). Let $T(\Delta)$ for $\Delta\in\mathcal{B}^{ach}$ be a nonnegative bounded operator on a complex separable Hilbert space $\mathcal{H}$. Suppose  $T(\emptyset)=0$ and  $\sum_nT(\Delta_n)=I$   (in the strong operator topology)  for every sequence $(\Delta_n)$ of mutually disjoint sets in $\mathcal{B}^{ach}$ such that $\bigcup_n\Delta_n$ is maximal achronal. Then the map $T$ is  called an  AL. 
If $W$ is a  representation of $\tilde{\mathcal{P}}$ on $\mathcal{H}$, then $T$ is a (W-){\bf covariant} AL when
 $$W(g)T(\Delta)W(g)^{-1}=T(g\cdot \Delta)\quad \mbox{for every $g \in \tilde{\mathcal{P}}$ and $\Delta \in\mathcal{B}^{ach}$}\:.$$
Note that $T$ is a {\bf positive operator valued measure} (POVM) on every maximal achronal set $\Lambda$. It is {\bf normalized} as $T(\Lambda)=I$.
\\

{\bf Causal logic.}    
The orthogonality  relation  in $\R^4$ 
 \begin{equation}\label{AOR}
\mathfrak{x}\perp\mathfrak{y}\quad  \Leftrightarrow\quad  \mathfrak{x}\ne\mathfrak{y} \textnormal{ and } (\mathfrak{x}-\mathfrak{y})^{\cdot 2}\le 0 \:.
\end{equation}
is called  \textbf{achronal separateness}.
Let $M\subset \R^4$. Then $M^\perp :=\{\mathfrak{x} \in \R^4 \: :\: \mathfrak{x} \perp \mathfrak{y}\: \text{ for all }\mathfrak{y}\in M  \}$ denotes  the {\bf complement}  and $ M^\land := (M^\perp)^\perp$ the  {\bf completion} of $M$ with respect to $\perp$.
$M$ is said to be  {\bf complete} if  $M =  M^\land$.
 According to \cite[Corollary 1]{CJ77}, $(\R^4,\perp)$ is a complete $D$-space, which means that every $M\subset \R^4$ which is 
  complete 
  is the completion of any maximal achronal set contained in $M$ (recall  \cite[(1)(c)]{C24}).
The set of \textbf{determinacy} of $M\subset\R^4$ with respect to $\perp$  is defined as
\begin{equation}\label{RDAR} 
M^{\sim}:=\{\mathfrak{x}:\forall\;\mathfrak{z} \textnormal{ with }\mathfrak{z}^{\cdot 2}>0\;\exists\, s\in\R \textnormal{ with } \mathfrak{x}+s\mathfrak{z}\in M\}
\end{equation}
It consists of  all points $\mathfrak{x}$ such that every timelike  straight line through $\mathfrak{x}$  meets $M$. $M\subset\R^4$ is said to be {\bf determining} if $M^\sim\setminus M \ne \emptyset$. 
  Now the decisive property is
\begin{equation}\label{AADAC}
  A \text{\,  achronal set } \Rightarrow  A^\sim=A^\land
 \end{equation}
 see   \cite[(35)]{C24}.  Finally, if $M$ is complete and $A$ maximal achronal in $M$, then $M$ is Borel if and only if $A$ is Borel  \cite[Lemma 4.1]{CJ77}.
For more details see \cite{CJ77} and   \cite[sec.\,24]{C17}.
The \textbf{causal logic} $\mathcal{C}$ is the bounded lattice of complete  {\em Borel} subsets of $\R^4$ which is partially ordered by  set inclusion $\subset$  and which admits achronal separateness 
 (\ref{AOR}) as orthocomplement operation. It is $\sigma$-complete, irreducible, orthomodular, and atomic.

{\bf Representation of the causal logic} (RCL).  Let $F(M)$ for  $M\in \mathcal{C}$  be a nonnegative bounded operator on $\mathcal{H}$. Suppose $F(\emptyset)=0$, 
and $\sum_nF(M_n)=I$  (in the strong operator topology) for every sequence $(M_n)$ of mutually orthogonal sets in $\mathcal{C}$ 
such that 
$\bigvee_nM_n=\R^4$. Then the map $F$ is  called an RCL. If $W$ is a representation of $\tilde{\mathcal{P}}$ on $\mathcal{H}$, then $T$ is a (W-){\bf covariant} RCL when
\begin{equation*}\label{PCCLL}
W(g)F(M)W(g)^{-1} = F(g\cdot M) \quad \mbox{for every $g \in \tilde{\mathcal{P}}$ and $M \in\mathcal{C}$}
\end{equation*}  

 For some properties of RCL see  appendix (\ref{PRCL}).

\section{Conserved probability density current for the Dirac system}

The following recalls briefly more or less well-known facts about the Dirac system.  For more details see \cite[sec.\,13-19]{C17}.

\textbf{Position space representation.} The Dirac system $W^\textsc{d}$ is  the direct sum  of  two irreducible  rep of $\tilde{\mathcal{P}}$ with mass $m>0$,  spin $\frac{1}{2}$,  and respectively positive and negative energy. It is a  particular representative

$$W^\textsc{d}\in[m,\text{\SMALL{$\frac{1}{2}$}},+]\oplus [m,\text{\SMALL{$\frac{1}{2}$}},-]$$

of the unitary equivalence class $[m,\text{\SMALL{$\frac{1}{2}$}},+]\oplus [m,\text{\SMALL{$\frac{1}{2}$}},-]$.
 In position space  $L^2(\R^3,\C^4)$  it reads 
 \begin{equation}\label{DSRPS}
 \big(W^\textsc{d}(g)\psi\big)(x)= s(A) \big(\operatorname{e}^{-\operatorname{i}y_0 H}\psi\big)(y)
\end{equation}
where $g=(\mathfrak{a},A)\in\tilde{\mathcal{P}}$, $(y_0,y):=g^{-1}\cdot (0,x)=A^{-1}\cdot (-a_0,x-a)$, and $H=\sum_{k=1}^3\alpha_k\,\frac{1}{\operatorname{i}} \partial_k+ \beta m$. We use the Weyl representation  
 $$\beta =\left( \begin{array}{cc} 0 & I_2\\ I_2 & 0 \end{array}\right),\quad \alpha_k =\left( \begin{array}{cc}  \sigma_k & 0\\ 0 & -\sigma_k \end{array}\right),\quad k=1,2,3 $$ 
with $\sigma_k$   the Pauli matrices, and $s(A)=\operatorname{diag}(A,A^{*-1})$. Moreover, $\sigma_0:=I_2$ and $\alpha_0:=I_4$.\\
\hspace*{6mm}
$H$ is selfadjoint. To be precise  $H=\mathcal{F}^{-1}H^{mom}\mathcal{F}$, where in the momentum representation $H^{mom}$ is the maximal multiplication operator by $$h:\R^3\to \C^4,\quad  h(p):=\sum_{k=1}^3\alpha_k\ p_k+ \beta m$$ 
The two-fold degenerate eigenvalues of $h(p)$ are $\pm\epsilon(p)$,  $\epsilon(p):=\sqrt{|p|^2+m^2}$.

Obviously $ \mathcal{S}(\R^3,\C^4)\subset  \operatorname{dom}H \cap  \operatorname{dom}H^{mom}$. We will use the 
 $W$-invariant dense subspace of the position space $$\mathcal{D}:=\mathcal{F}^{-1}\big(C_c^\infty(\R^3,\C^4)\big) \subset \mathcal{S}(\R^3,\C^4)$$

\textbf{Time-dependent wave functions.}
 The time-dependent Dirac wave-function $\Psi$ is the solution of  the  initial-value problem for  the \textbf{Dirac equation} 
\begin{equation} 
i\partial_0\Psi = H\Psi, \quad
\Psi(0,\cdot)=\psi \text{ on } \R^3 \text{ with } \psi\in \operatorname{dom} H
\end{equation}
Hence $\psi$ and $\Psi$ are related by $\Psi(\mathfrak{x})=\big(\operatorname{e}^{-\operatorname{i} x_0 H}\psi\big)(x)= \big(W^\textsc{d}(-x_0)\psi\big)(x)$.

From this relation and  (\ref{DSRPS}) one  gets 
the transformation law under Poincar\'e transformations  for $\Psi$.  Indeed, let $\Psi_g$ be the wave-function with initial value $W^\textsc{d}(g)\psi$. Then 
\begin{equation}\label{TTDWF}
 \Psi_g(\mathfrak{x})=s(A)\, \Psi(g^{-1}\cdot\mathfrak{x})
\end{equation}
 since
$W^\textsc{d}(-x_0)W^\textsc{d}(g)= W^\textsc{d}(g')$ for $g':=(-x_0,0;I_2)g$ satisfying $g'^{-1}\cdot(0,x)=g^{-1}\cdot \mathfrak{x}$.
\\

\textbf{Conserved Current.}
The  Dirac probability density current $\mathfrak{J}=(J_i)$ is related to the time-dependent wave-function in the usual manner
 $J_i:=\Psi^*\alpha_i\Psi=\sum_{l=1}^4\overline{\Psi}_l  (\alpha_i\Psi)_l$, $i=0,1,2,3$. We use also  the notation adapted to our needs
\begin{equation}\label{DSPC}
 J_i(\psi,\mathfrak{x})=\langle \big(\e^{-\i x_0H}\psi\big)(x),\alpha_i\, \big(\e^{-\i x_0H}\psi\big)(x)\rangle_{\C^4}
\end{equation}
where $\psi\in L^2(\R^3,\C^4)$, $\mathfrak{x}=(x_0,x)\in\R^4$, $i=0,1,2,3$.

\begin{Lem}\label{PDPC}
The following properties 
$\emph{(a)-(f)}$ of the Dirac probability density current $\mathfrak{J}$ hold. Let $\psi\in L^2(\R^3,\C^4)$.
\begin{itemize}
\item[\emph{(a)}]   $\mathfrak{J}$ is real
\item[\emph{(b)}]  $J_0(\psi;0,x)=|\psi(x)|^2\ge 0$
for every $x\in\R^3$ and  $\int J_0(\psi;0,x)\d^3 x=||\psi||^2$
\item[\emph{(c)}]   $\mathfrak{J}$  is  Poincar\'e covariant, i.e.,  \,$\mathfrak{J}(W^\textsc{d}(g)\psi,\mathfrak{x}\big)=A\cdot \mathfrak{J}(\psi,g^{-1}\cdot\mathfrak{x})$ for $g\in \tilde{\mathcal{P}}$
\item[\emph{(d)}]  $J_0(\psi,\mathfrak{x})\ge |J(\psi,\mathfrak{x})|$ for every $\mathfrak{x}\in\R^4$
\item[\emph{(e)}]  $\mathfrak{J}$ is  $C^\infty$ and satisfies the continuity equation $\sum_{i=0}^4\partial_iJ_i(\psi,\mathfrak{x})=0$ for  $\mathfrak{x}\in\R^4$,  $\psi\in  \mathcal{S}(\R^3,\C^4)$
\item[\emph{(f)}]  $\mathfrak{J}(\psi,\cdot)$ is bounded for $\psi\in\mathcal{D}$
\end{itemize}
\end{Lem}

{\itshape Proof.} (a) The matrices $\alpha_k,\beta$ are self-adjoint.       (b)  is obvious.  (c) $\mathfrak{J}(W^\textsc{d}(g)\psi,\mathfrak{x}\big)=\big(\langle \Psi_g(\mathfrak{x}),\alpha_i \Psi_g(\mathfrak{x})\rangle\big)=\big( \langle s(A)\Psi(g^{-1}\cdot \mathfrak{x}),\alpha_i s(A)\Psi(g^{-1}\cdot \mathfrak{x})\rangle \big)=\big( \langle \Psi(g^{-1}\cdot \mathfrak{x}), s(A)^*\alpha_i s(A)\Psi(g^{-1}\cdot \mathfrak{x})\rangle \big)$ using (\ref{TTDWF}). The claim follows since $$s(A)^*\alpha_i s(A)=\sum_{j=0}^3\Lambda(A)_{ij}\alpha_j, \quad i=0,1,2,3$$ holds for $A\in SL(2,\C)$. This is due to the fundamental relation  valid for all two by two matrices $A$
\begin{equation}
 A\sigma_i A^*=\sum_{j=0}^3\Lambda(A)_{ji}\sigma_j,  \quad i=0,1,2,3\tag{*}
 \end{equation} 
(d) As  $J_0(\psi;0,x)\ge 0$ for all $x$,$\psi$, one has $\mathfrak{J}(W^\textsc{d}(g)^{-1}\psi;0,x)\cdot (1,0,0,0)\ge 0$ for  $g=(\mathfrak{a},A)\in\tilde{\mathcal{P}}$, whence $\mathfrak{J}(\psi; \mathfrak{a}+A\cdot(0,x))\cdot (A\cdot(1,0,0,0)\ge 0$ by covariance  and hence $\mathfrak{J}(\psi; \mathfrak{x})\cdot \mathfrak{e}\ge 0$ for $\mathfrak{x}\in\R^4$,  $\mathfrak{e}^{\cdot2} =1$, $e_0>0$. The claim follows.
(e)   Note that $\Psi \in C^\infty(\R^4; \C^4)$ if $\varphi \in \mathcal{S}(\R^3,\C^4)$.  Indeed, as $\Psi$ is the Fourier transform of  $\R^3\ni  p\mapsto  e^{-ix_0 h(p)}\varphi(p)$, passing the $x_0$ and $x$ derivatives  under the sign of $d^3p$ integration the result follows by standard arguments.
Hence   $\mathfrak{J}$ is  $C^\infty$ and
$ \partial_0(\Psi^*\Psi) =(\partial_0 \Psi)^* \Psi+ \Psi^*\partial_0\Psi=\big(-\sum_{k=1}^3\partial_k\Psi^*\alpha_k- \frac{1}{\i} \Psi^*\beta m\big)\Psi+\Psi^*\big(-\sum_{k=1}^3\alpha_k\partial_k\Psi+\frac{1}{\i}\beta m\Psi\big)=-\sum_{k=1}^3\partial_k(\Psi^*\alpha_k\Psi)$, whence the claim. (f) Due to (d) it suffices to show that $J_0(\psi,\cdot)$ is bounded:  $(2\pi)^{3/2} |\big(\operatorname{e}^{-\operatorname{i} x_0 H}\psi\big)(x)|=  |\int \e^{\i px}\e^{-\i x_0 h(p)}\phi(p) \d^3p|\le \int |\phi(p)|  \d^3p<\infty $ as $\phi:=\mathcal{F}\psi\in C_c^\infty$. 
\qed 

\section{Integrable majorant regarding the current}\label{IMRC}
The following estimate derived by the non-stationary phase method is an application of  \cite[Theorem 1.8]{T92} in \cite[(77) Theorem]{C17}. For  (\ref{FEATRE}) cf.\,also \cite[Theorem XI.17(a)]{RS79}.
 It  provides an integrable majorant of the current on maximal achronal sets.

\begin{Lem} \label{FEATRE}Let $f\in C_c^\infty(\R^3)$ and $g(\mathfrak{x}):=\int \e^{\i (px+x_0\epsilon(p))} f(p)\,\d^3p$ for $\mathfrak{x}\in\R^4$.  Then for some $0<\gamma<1$ and every $0<N<\infty$ there exists a constant $0<A_N<\infty$ such that 
$$ |g(\mathfrak{x}) |\le  A_N(1+|x|+|x_0|)^{-N}  \quad \text{ for } |x|\ge  \gamma |x_0|$$
\end{Lem}\\
{\itshape Proof.} Let $K:=\operatorname{supp}(f)$. Put $\beta:=\max\{\frac{|p|}{\epsilon(p)}:p\in K\}$. Clearly $0\le\beta<1$. Introduce $\omega(p):= (|x|+|x_0|)^{-1}\big(px+x_0\epsilon(p)\big)$. Then $\nabla \omega(p)=  (|x|+|x_0|)^{-1}\big(x+\frac{x_0}{\epsilon(p)}p\big)$ and $|\nabla \omega(p)|\ge  (|x|+|x_0|)^{-1}\big(|x|-|x_0|\frac{|p|}{\epsilon(p)}\big)\ge \frac{|x|-\beta |x_0|}{|x|+|x_0|}$ for $p\in K$. Now let $\beta <\gamma <1$ and 
suppose $|x|\ge \gamma|x_0|$. Then $|\nabla \omega(p)|\ge  \frac{|x|-\beta|x|/\gamma}{|x|+|x|}=\frac{1-\beta/\gamma}{2}>0$ for $p\in K$. Note also that the derivatives satisfy 
$|D^\alpha \omega(p)|\le 1$ for $|\alpha|=1$, $|D^\alpha \omega(p)|\le |D^\alpha \varepsilon(p)|$ for $|\alpha|\ge 2$. Hence integration by parts as in the proof of  \cite[Theorem 1.8]{T92} yields  the claim.\qed
\\

\hspace*{6mm}
 Let $\psi\in \mathcal{D}$ and put $\varphi:=\mathcal{F}\phi$. In momentum space representation  time evolution is $\varphi_{x_0}:=\operatorname{e}^{\operatorname{i}x_0h}\varphi$, i.e., $\varphi_{x_0}(p)=\operatorname{e}^{\operatorname{i}x_0h(p)}\varphi(p) \,\forall\; p$.  Let $\eta\in\{+,-\}$. Then 
\begin{equation}\label{PEP}
 \pi^\eta(p)=\frac{1}{2}\left(I+\frac{\eta}{\epsilon(p)}h(p)\right)
 \end{equation}
  is the projection in $\C^4$ onto the $2$-dimensional eigenspace of $h(p)$ with eigenvalue $\eta\, \epsilon(p)$.   So $\varphi^\eta:=\pi^\eta\varphi$, $\eta=\pm$, 
are the projections  of $\varphi$ onto the   eigenspaces of $\operatorname{sgn}(H^{mom})$. Still $\varphi^\eta     \in C_c^\infty$.  Analogously 
 $(\varphi_{x_0})^\eta:=\pi^\eta\varphi_{x_0}$. Note that $(\varphi_{x_0})^\eta=(\varphi^\eta)_{x_0}=\operatorname{e}^{\operatorname{i}x_0\eta \epsilon}\varphi^\eta$, as $\operatorname{e}^{\operatorname{i}x_0h}$ and $\pi^\eta$ commute. Define analogously $\psi^\eta$ and $\psi_{x_0}^\eta$. Note $\psi_{x_0}^\eta =      \mathcal{F}^{-1} \big(\e^{\i x_0\eta \epsilon}\varphi^\eta\big)$, $(\psi_{x_0}^\eta)_l =      \mathcal{F}^{-1} \big(\e^{\i x_0\eta \epsilon} (\varphi^\eta)_l\big)$  for each component $l=1,2,3,4$, and hence
 \begin{equation}\label{ENSPM}
 (\psi^\eta_{x_0}(x))_l=(2\pi)^{-3/2}\int \operatorname{e}^{\operatorname{i}(px+x_0\eta\epsilon(p))}(\varphi^\eta(p))_l\,\operatorname{d}^3p
 \end{equation}
with $ |\psi_{x_0}(x)|^2=|\sum_\eta  \psi^\eta_{x_0}(x)|^2\le 2\sum_\eta  |\psi^\eta_{x_0}(x)|^2=2 \sum_{l,\eta} | (\psi^\eta_{x_0}(x))_l|^2$.

\begin{Pro} \label{NER} 
Let $\psi\in \mathcal{D}$.  
Then for every $0<N < \infty$ there is a constant  $0<C_N<\infty$  depending  on $\psi$ such that $$J_0(\psi,\mathfrak{x})\le C_N(1+|x|)^{-N}\quad \text{ for } |x|\ge |x_0|$$
\end{Pro}\\
{\itshape Proof.}   Note  $J_0(\psi,\mathfrak{x})=|\psi_{-x_0}(x)|^2$.  By (\ref{ENSPM}),\,(\ref{FEATRE}),  $J_0(\psi,\mathfrak{x})\le 2 \sum_{l,\eta} | (\psi^\eta_{-x_0}(x))_l|^2 \le 2\sum_{l,\eta}A_N^{l,\eta}(1+|x|+|x_0|)^{-N}$, whence the claim.\qed

\section{Covariant AL  of the Dirac system and  of the  electron}\label{CALDSE} 
Now we are in a position to apply  \cite[Theorem 19]{CDM25} thus achieving   our main result.

\begin{The}\label{MTDC}
 Let $\mathfrak{J}$ be the Dirac probability density current. Then there is a unique  AL $T^\textsc{d}$ of the Dirac system satisfying for every achronal Borel set $\Delta$ and  $\psi\in \mathcal{D}$
\begin{equation}\label{REXAL}
\langle \psi, T^\textsc{d}(\Delta)\psi\rangle= \int_{\varpi(\Delta)} \big( J_0(\psi;\tau(x),x) -  J(\psi;\tau(x),x)\nabla\tau(x)\big)  \d^3x 
\end{equation}
where $\tau:\varpi(\Delta)\to \R$ with $\operatorname{graph}\tau=\Delta$. One has the covariance $$W^\textsc{d}(g)T^\textsc{d}(\Delta)W^\textsc{d}(g)^{-1}=T^\textsc{d}(g\cdot \Delta)$$
\end{The}
\\
{\itshape Proof.}  Due to (\ref{PDPC}),\,(\ref{NER}), the general result  \cite[Theorem 19]{CDM25} yields the claim.\qed
\\

 Let $\mathcal{H}^\textsc{d}$ denote the state space  of the Dirac system. Recall that the invariant eigenspace of $\operatorname{sgn}(H)$ for the eigenvalue $+1$ is considered to be  the state space of the {\it electron}, hence denoted by $\mathcal{H}^e$.
 In  momentum representation $\mathcal{H}^e$ is the range of the orthogonal projection $\varphi \mapsto \pi^+\varphi$ given by (\ref{PEP}). Let $j_e: \mathcal{H}^e \to  \mathcal{H}^\textsc{d}$ 
 be the identical injection. Then the \textbf{trace }(compression) $T^e$ of $T^\textsc{d}$   in (\ref{MTDC}) on $ \mathcal{H}^e $ is defined as
 \begin{equation}\label{ALDE}
 T^e(\Delta):=j_e^*\, T^\textsc{d}(\Delta)\,j_e,  \quad \Delta\in\mathcal{B}^{ach}
\end{equation}

Obviously on has
  
\begin{Cor}\label{EMTDC}
$T^e$ $\emph{(\ref{ALDE})}$ is a covariant AL of the electron. Covariance holds by means of the  electron representation being the  representation $W^e:=j_e^*\,W^\textsc{d}\,j_e$  on $ \mathcal{H}^e$. Due to
\begin{equation}\label{EPL} 
\langle \psi, T^e(\Delta)\psi\rangle=\langle \psi, T^\textsc{d}(\Delta)\psi\rangle, \quad  \Delta\in\mathcal{B}^{ach}, \; \psi\in\mathcal{H}^e
\end{equation}
the integral representation  \emph{(\ref{IREXAL})} for $T^\textsc{d}$ holds  for  $T^e$ with respect to $\mathfrak{J}^e$ given by 
$$\mathfrak{J}^e(\psi,\mathfrak{x}):=\mathfrak{J}(\psi,\mathfrak{x}), \;\; \psi \in \mathcal{D}^e:=\mathcal{F}^{-1}\big(\pi^+ C_c^\infty(\R^3,\C^4)\big) $$
\end{Cor}

\section{Covariant RCL for the  Dirac system and the  electron} \label{CRCLDSE}
One recalls the {\it one-to-one correspondence} of (covariant) AL and (covariant) RCL as expounded in \cite[(19),\,(20)]{C24} and reported in \cite[(27)]{CDM25}. Accordingly there is a unique RCL $F^\textsc{d}$ for the Dirac system such that  
\begin{equation}
 F^\textsc{d}(\Delta^\land)=T^\textsc{d}(\Delta), \quad \Delta\in\mathcal{B}^{ach}
 \end{equation}
 $F^\textsc{d}$ is covariant by means of   $W$. 
 \\  
\hspace*{6mm}
Analogously there is a unique RCL $F^e$ for the electron   satisfying 
\begin{equation}\label{MTCL}
 F^e(\Delta^\land)=T^e(\Delta),   \quad \Delta\in\mathcal{B}^{ach}
  \end{equation} 
 $F^e$ is covariant be means of $W^e$. Obviously $F^e=j_e^*\,F^\textsc{d}\,j_e$ holds. 
 \\  
\hspace*{6mm}
Thus  a long outstanding achievement is reached.

\section{Properties of the AL of the Dirac system}\label{PALDS}

Several physically relevant properties of the AL $T^\textsc{d}$ for the Dirac system constructed in (\ref{MTDC}) and the related AL $T^e$  for the electron are 
shown.

 \subsection{Dirac and Newton-Wigner position operator}\label{DNWPO}
Let $\Delta$  be a Borel subset  of the Euclidean space $\varepsilon:=\{\mathfrak{x}\in\R^4: x_0=0\}\equiv \R^3$. By (\ref{MTDC}) one has $\langle \psi, T^\textsc{d}(\Delta)\psi\rangle= \int_\Delta|\psi(x)|^2\d^3x$ for all $\psi\in\mathfrak{D}$. So $T^\textsc{d}(\Delta)$ is the multiplication operator by  the indicator function $1_\Delta$ of $\Delta$, and  $T^\textsc{d}$ is the canonical  projection operator valued measure on $\varepsilon$. The corresponding position vector operator $X^\textsc{d}$ is the multiplication operator by the identity function $\operatorname{id}_{\R^3}$.  $X^\textsc{d}$  does not preserve the electron subspace  $\mathcal{H}^e$. This is shown by the following obvious lemma.

\begin{Lem}\label{DPOCSH} In momentum representation $X^\textsc{d}$ equals  $\i\nabla$ and $\operatorname{sgn}(H)$ is the multiplication operator by $\frac{1}{\epsilon}h$.  So in momentum representation $[X^\textsc{d},\operatorname{sgn}(H)]$ is the multiplication operator by $\i\nabla (\frac{1}{\epsilon}h)$, which does not vanish.
\end{Lem}

Therefore $(W^\textsc{d},X^\textsc{d})$  is irreducible. Henceforth  by {\it Dirac system} we  refer more precisely  to $(W^\textsc{d},T^\textsc{d})$.
 \\  
\hspace*{6mm}
 Apart from massless helicity zero systems not considered here,  the  Newton-Wigner position operator $X^\textsc{nw}$ \cite{NW49},\,\cite{W62} is defined for all massive systems. A massive  system is described  by  a rep $W$ of $\tilde{\mathcal{P}}$, for  which the  mass-squared  operator $C:=H^2-P^2$ is strictly positive, i.e., $C>0$. It is characteristic for  $X^\textsc{nw}$ that it commutes with  the Casimir operators $C$, the Pauli-Lubanski  scalar,
 and  the {\it sign of the energy}. This is a consequence of (\ref{BTFX}) below. Hence $(W^\textsc{d},X^\textsc{nw})$ is reducible. Therefore $X^\textsc{d} \ne X^\textsc{nw}$.
  \\  
\hspace*{6mm}
  The remarkable Bakamjian-Thomas-Foldy formula  
    \begin{equation}\label{BTFX}
X^{\textsc{nw}} =\textrm{\SMALL{$\frac{1}{2}$}}(H^{-1}N+NH^{-1})-C^{-1/2}(C^{1/2}+|H|)^{-1} \,P\times \big(J+H^{-1} (P\times N)\big)
\end{equation}
 holds on a common core and may be regarded as the defining equation for the Newton-Wigner position operator  as it expresses  $X^{\textsc{nw}}$ in terms of the generators of  $\tilde{\mathcal{P}}$.

\begin{The}\label{RNWD} Regarding the Dirac system there is a bounded skew-adjoint vector operator $F=(F_j)_{j=1,2,3}$ on $\mathcal{H}^\textsc{d}$   such that 
\begin{equation}\label{NWODO}
X^\textsc{nw}=  X^\textsc{d} + \i\,F 
\end{equation}
In momentum representation $F_j$ is the multiplication operator by $\Phi_j:\R^3\to \C^{4\times 4}$,
$$\Phi_j(p):= \frac{1}{2\epsilon(p)(\epsilon(p)+m)}     \left( \begin{array}{cc}  -p_j+ \slashed{p}\, \sigma_j & -(\epsilon(p)+m)\sigma_j +\frac{1}{\epsilon(p)}p_j\,\slashed{p}\\  (\epsilon(p)+m)\sigma_j -\frac{1}{\epsilon(p)}p_j\,\slashed{p} &  -p_j+ \slashed{p}\, \sigma_j  \end{array}\right)$$
with $\slashed{p}:=\sum_j\sigma_jp_j$.
\end{The}

{\it Proof.} One gains $X^{\textsc{nw}}$ evaluating (\ref{BTFX}) for $W=W^\textsc{d}$ in momentum representation. For this one substitutes $N$ by $\frac{1}{2}(HX^\textsc{d}+X^\textsc{d}H)$. This well-known relation is easy  to verify.  Note $N_j^{mom}=\frac{1}{2}\alpha_j+h\, \i\partial_j$. Moreover note $J=L-P\times X^\textsc{d}$ with 
$L_j^{mom}=\frac{1}{2}\operatorname{diag}(\sigma_j,\sigma_j)$. Then the result follows by a  straight forward computation.\qed
\\

 The Newton-Wigner localization $(W^\textsc{d},T^\textsc{NW})$ determined by $X^{\textsc{nw}}$ and the  Dirac localization $(W^\textsc{d},T^\textsc{d})$ confined to the Euclidean space $\varepsilon$ determine two systems of imprimitivity for the Euclidean group $\tilde{\mathcal{E}} \subset \tilde{\mathcal{P}}$ acting on $\varepsilon$. Hence by  the inprimitivity theorem there is a unitary transformation $Y$ on  $\mathcal{H}^\textsc{d}$ which commutes with $W^\textsc{d}|_{\tilde{\mathcal{E}}}$ and satisfies $T^\textsc{NW}(\Delta) =Y^{-1} T^\textsc{d}(\Delta)Y $ for every Borel $\Delta\subset \varepsilon$.

\begin{Lem}\label{MPY} 
For $p\in\R^3$ put
\begin{center}
$\Upsilon(p):=\frac{1}{2\sqrt{\epsilon(p)(\epsilon(p)+m)}} \left( \begin{array}{cc} \epsilon(p)+m+\slashed{p} &  \epsilon(p)+m-\slashed{p} \\  \epsilon(p)+m-\slashed{p}  &  -\epsilon(p)-m-\slashed{p}  \end{array}\right)$
\end{center}
    $\Upsilon(p)$ is unitary  self-adjoint,  $\Upsilon(p)^{-1}=\Upsilon(p),\, \frac{1}{\epsilon(p)}h(p)=\Upsilon(p)\operatorname{diag}\big(I_2,-I_2\big)\Upsilon(p)^{-1}$, and
$\Upsilon(B\cdot p)=\operatorname{diag}(B,B)\Upsilon(p)\operatorname{diag}(B^*,B^*)$ for $B\in SU(2)$. 
\end{Lem}

{\it Proof}. Use $\slashed{p}^2=p^2$, formula (*) in the proof of (\ref{PDPC}), and note $h(p)=\left( \begin{array}{cc} \slashed{p} & m\\ m & -\slashed{p}  \end{array}\right)$.\qed

It is not difficult to construct $\Upsilon(p)$ having the above properties.

\begin{The} Let $Y$ be the operator on $\mathcal{H}^\textsc{d}$ which in momentum representation is the multiplication operator by $\Upsilon$ $$Y^{mom}\varphi (p):=\Upsilon(p)\varphi(p)$$
 Then \emph{(i)} $Y$ is unitary and self-adjoint,   $Y^{-1}=Y$,  \emph{(ii)}    $Y$   and   $W^\textsc{d}|_{\tilde{\mathcal{E}}}$  commute and  \emph{(iii)}   $Y^{-1}X^\textsc{d}Y$ and $\operatorname{sgn}(H)$ commute. Finally 
 \begin{equation}\label{UANWD}
 X^\textsc{nw}= Y^{-1} X^\textsc{d} Y
\end{equation}
\end{The}
\\
{\it Proof.}  The claims (i) - (iii) hold by (\ref{MPY}). As to (ii) note $W^{\textsc{d}\,mom}(b,B)\varphi (p)=\e^{-\i bp}\operatorname{diag}(B,B)\varphi(B^{-1}\cdot p)$.  (ii) holds since $X^\textsc{d}$ and $Y\operatorname{sgn}(H)Y^{-1}$ commute as in momentum representation $X^\textsc{d}$ equals $\i\nabla$ and $Y\operatorname{sgn}(H)Y^{-1}$ is a multiplication operator by a constant matrix.
 \\  
\hspace*{6mm}
For the proof of  (\ref{UANWD}) we evaluate $Z:= Y^{-1} X^\textsc{d} Y$ in momentum representation. One readily obtains $Z^{mom}\varphi (p)=(\i \nabla\varphi)(p) +\Upsilon^{-1}(p)(\i\nabla \Upsilon)(p) \varphi(p)$. A straight forward computation shows 
\begin{equation}\label{EE}
 \Phi = \Upsilon^{-1}(\nabla \Upsilon)
\end{equation}
 where   $ \Phi$ is introduced in (\ref{RNWD}). Hence $Z^{mom}=X^{\textsc{nw}\,mom}$.\qed

{\bf Remark.} In \cite{NW49},\,\cite{W62}  uniqueness of $X^\textsc{nw}$ for an irreducible massive system is achieved adding invariance under time inversion and a regularity assumption. Regarding $Z^{mom}$ in the proof of (\ref{UANWD}) the latter holds since $\Upsilon$ is smooth. Also the former $\mathcal{T}Z^{mom}\mathcal{T}^{-1}=Z^{mom}$ holds  true for the time inversion transformation $\mathcal{T}$ (\ref{TIO}). For the verification use $\overline{\slashed{p}}=\i \sigma_2 \slashed{p} \i\sigma_2$.
Since $\mathcal{T}$ and $\operatorname{sgn}(H)$ commute, too, (\ref{UANWD}) follows without any further computation.
 
In contrast to the Dirac $X^\textsc{d}$ the Newton-Wigner  $X^\textsc{nw}$ is not causal. As well-known the Newton-Wigner localization is frame-depedent. In favor of Dirac's position operator is the further  fact  that in case of minimal coupling it indicates the point of interaction of electromagnetic field and particle.  It should be hard to estimate the changes of the fine structure of atomic spectra when Dirac's position operator is replaced by Newton-Wigner's one according to (\ref{NWODO}),\,(\ref{EE}).
 \\  
\hspace*{6mm}
 However it is easy to estimate the deviation of the expectation values with respect to the electron states of the Newton-Wigner position  from those  of the Dirac position, see (\ref{EEDEV}). Note that obviously the eigenvalues of $\i\Phi^e_j(p)$ are smaller than those of  $\i\Phi_j(p)$.  It is interesting to quantify the improvement of the estimate 
in (\ref{EEDEV}) using $\lambda_j^e(p)$ in place of $\lambda_j(p)$. See appendix \ref{EVP}.
 \\  
\hspace*{6mm}
Recall (\ref{PEP}) for $\pi^+(p)$. In momentum representation  $\pi^+\varphi$ is the electron part of the Dirac state $\varphi$.

\begin{Lem}\label{ECEE}
Let $p\in\R^3$, $j=1,2,3$, and write $\epsilon$ for $\epsilon(p)$. Put $$\Phi^e_j(p):=\pi^+(p)\Phi_j(p) \pi^+(p),\quad \lambda_j^e(p):=\frac{1}{2\epsilon(\epsilon+m)}\sqrt{p^2-p^2_j}$$ Then
\begin{itemize}
\item[\emph{(i)}] $\Phi^e_j(p) = -\frac{1}{4\epsilon^2(\epsilon+m)} \left(\begin{array} {cc}  -p_j+ \slashed{p}\, \sigma_j & 0\\ 0 &-p_j+ \slashed{p}\, \sigma_j \end{array}\right)\left( \begin{array}{cc}m & \epsilon-\slashed{p} \\  \epsilon+\slashed{p} &m \end{array}\right)$
\item[\emph{(ii)}] 
If $p^2=p_j^2$ then $\Phi^e_j(p)=0$. If $p^2\ne p_j^2$ then the nonzero  eigenvalues of $\i\Phi^e_j(p)$ are simple and equal\, $\pm \lambda_j^e(p)$. 
\item[\emph{(iii)}] Let  $\varphi:=\frac{1}{2}(1+\sqrt{5})$ be the golden ratio. Then
$$ \lambda_j^e(p)\le \frac{|p|}{2\epsilon(\epsilon+m)}\le ||\lambda_j^e||_\infty =   \frac{\sqrt{\varphi}}{4\varphi+2}m^{-1}\approx 0,15\, m^{-1}$$
where the second $\le$ is an equality iff $|p|=\sqrt{\varphi}\,m$. Also   $\lambda_j^e(p)\le \frac{1}{2|p|}$.
\end{itemize}
\end{Lem}

{\it Proof.} (i) follows by straight forward computation. (ii) Obviously the electron state  $\varphi=\pi^+\varphi$ satisfies equivalently $\varphi=(m\chi, (\epsilon-\slashed{p})\chi)$ with $\chi \in L^2(\R^3,\C^2)$. Hence the eigenvalue equation reads
$$\i\Phi_j^e(p)\left(\begin{array} {c} m\,x \\ (\epsilon-\slashed{p})x \end{array}\right) =\lambda \left(\begin{array} {c} m\,x \\ (\epsilon-\slashed{p})x \end{array}\right) $$
for $x\in\C^2$, $\lambda\in\R$. It is equivalent to 
\begin{equation}
\i (\epsilon+\slashed{p})( -p_j+ \slashed{p}\, \sigma_j )x=\mu \,x\tag{*}
\end{equation}
for $\mu :=-2\epsilon m(\epsilon+m)\lambda$. Put $A:=-p_j+ \slashed{p}\, \sigma_j$. 
Now verify that $\operatorname{Tr}(A)=0$ and  $\det(A)=p^2-p_j^2 $. Hence  $\operatorname{Tr}(\i(\epsilon+\slashed{p})A)=\i\operatorname{Tr}(\epsilon A -p_j\slashed{p}+p^2\sigma_j)=0$, and $\det(\i (\epsilon+\slashed{p})A)=-\det(\epsilon+\slashed{p})\det(A)=-m^2(p^2-p_j^2)$. Therefore $\mu=\pm m\sqrt{p^2-p_j^2}$, whence $\lambda=\pm \lambda_j^e(p)$, and the claim holds. (iii) is easily  verified.\qed

\begin{Cor} Let $\i F^e$ be the trace of $\i F$ on the state space of the electron. Then $||\i F_j^e||<0,15 \,m^{-1}$, $j=1,2,3$. Similarly $X^e$ denotes the trace of $X^\textsc{d}$. Then \begin{equation}\label{EEDEV}
|\langle \psi, X_j^\textsc{nw}\psi \rangle - \langle \psi, X_j^e \psi \rangle| \le \int \lambda^e_j(p)|\varphi(p)|^2 \d^3p\le ||\lambda_j^e||_\infty
\end{equation}
holds for the electron state $\psi\in\operatorname{dom}(X_j^e)$, $||\psi||=1$, $\varphi=\mathcal{F}\psi$. (Here $X^\textsc{nw}$ denotes the NW-position operator regarding the electron.)
\end{Cor}

The question is whether the difference $\varDelta$ of the  electron expectation values  regarding the Dirac and the Newton-Wigner position operator is experimentally ascertainable. For a rough consideration we insert the constants $\hbar$ and $c$. According to (\ref{ECEE})  the relevant lengths are $\frac{\hbar}{|p|}$ and the Compton wavelength $\frac{\hbar}{mc}$.   By (\ref{EEDEV}), $\varDelta$ is smaller than the wavelength of the electron. 
\\  
\hspace*{6mm}
Consider an electron state $\varphi$. Let the density $|\varphi(p)|^2$ be concentrated around the origin.  It implies that the position of the electron  is indetermined. As the wavelength sets the minimum space required for the free electron \cite{G90}, $\varDelta$ is not ascertainable.
\\  
\hspace*{6mm}
Conversely, if the electron is well localized around some point, then the density $|\varphi(p)|^2$ is concentrated at large values of $|p|$. As $\lambda^e_j(p)<\frac{\hbar}{|p|}$, according to  (\ref{EEDEV}) the difference $\varDelta$ is considerably smaller than the wavelength and hence presumably not ascertainable.

\subsection{Projection operator valuedness} \label{POVDL} 
By Poincar\'e covariance it follows that $T^\textsc{d}(\Delta)$ is a projection operator for every spacelike flat Borel subset $\Delta$ of spacetime. 
Here one has  the following  general result (\ref{MTPVAL}) on AL.  It does not assume covariance nor the existence of a conserved current determining the AL.

 \begin{The}\label{MTPVAL} Let $T$ be an AL assigning  a projection operator to every  spacelike flat Borel set. Assume $T(\Delta)=0$ if $\mathcal{L}^3(\varpi(\Delta))=0$ for $\Delta\in\mathcal{B}^{ach}$.
 Then $T$ is projection operator valued.
\end{The}

 The proof of (\ref{MTPVAL}) is rather lengthy and is postponed to the appendix (\ref{P7T}).
There are several  preparatory results, some of them are interesting in their own right as e.g. those in appendix $A$. (\ref{MTPVAL})  obviously
  implies the objective (\ref{DSLPV}).

\begin{Cor}\label{DSLPV}  The AL $T^\textsc{d}$ of the Dirac system  \emph{(\ref{MTDC})}  is projection valued. 
\end{Cor}

\subsection{Causality condition, microscopic causality}\label{CCMC}

Being an AL,  $T^\textsc{d}$ satisfies the causality condition CC below.  The latter states that the probability of localization in {{\em any region of influence}}  determined by the limiting velocity of light is not less than that in the region of actual localization. {More precisely, consider  a region $\Delta'$ which is contained in a causal base $\Sigma'$.
Then $\Delta'$ is a {\bf  region of influence} of the achronal region  $\Delta$} if all causal straight  lines, which intersect  $\Delta$,  meet  $\Delta'$.
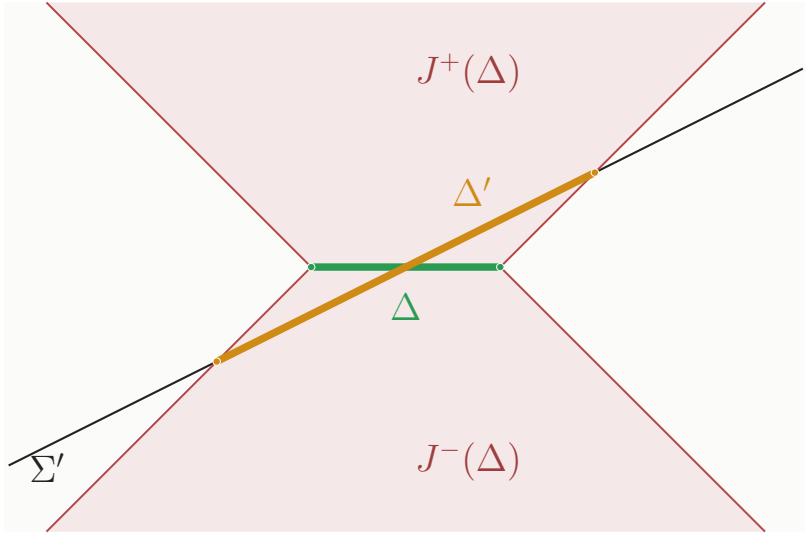
\begin{figure}[H]
	\centering
	
	\begin{tikzpicture}[x=1.25cm,y=1.25cm]
		
		\fill[backgroundcolor]
		(-4.25,-2.8) rectangle (4.25,2.8);
		
		\fill[causalcolor,opacity=0.10]
		(-3.8,2.8)
		-- (3.8,2.8)
		-- (1,0)
		-- (3.8,-2.8)
		-- (-3.8,-2.8)
		-- (-1,0)
		-- cycle;
		
		
		\draw[surfacecolor,line width=0.8pt]
		(-4.2,-2.1) -- (4.2,2.1);
		
		\draw[causalcolor,line width=0.75pt]
		(-3.8,2.8) -- (-1,0);
		
		\draw[causalcolor,line width=0.75pt]
		(-3.8,-2.8) -- (-1,0);
		
		\draw[causalcolor,line width=0.75pt]
		(1,0) -- (3.8,2.8);
		
		\draw[causalcolor,line width=0.75pt]
		(1,0) -- (3.8,-2.8);
		
		\draw[deltacolor,line width=2.8pt,line cap=round]
		(-1,0) -- (1,0);
		
		\draw[deltaprimecolor,line width=2.8pt,line cap=round]
		(-2,-1) -- (2,1);
		
		\filldraw[deltacolor,draw=white,line width=0.3pt]
		(-1,0) circle (1.4pt);
		
		\filldraw[deltacolor,draw=white,line width=0.3pt]
		(1,0) circle (1.4pt);
		
		\filldraw[deltaprimecolor,draw=white,line width=0.3pt]
		(-2,-1) circle (1.4pt);
		
		\filldraw[deltaprimecolor,draw=white,line width=0.3pt]
		(2,1) circle (1.4pt);
		
		
		\node[
		anchor=west,
		text=surfacecolor,
		font=\Large
		] at (-4.08,-2.12) {$\Sigma'$};
		
		\node[
		anchor=north,
		text=deltacolor,
		font=\Large
		] at (0,-0.18) {$\Delta$};
		
		\node[
		anchor=south,
		text=deltaprimecolor,
		font=\Large
		] at (0.70,0.55) {$\Delta'$};
		
		\node[
		anchor=west,
		text=causallabel,
		font=\Large
		] at (0.00,2.05) {$J^+(\Delta)$};

                   \node[
		anchor=west,
		text=causallabel,
		font=\Large
		] at (0.00,-2.05) {$J^-(\Delta)$};
		
	\end{tikzpicture}
	
	\caption{The achronal region $\Delta\subset\Sigma$ and its region of
		influence $\Delta'=\Sigma' \cap( J^+(\Delta)\cup J^-(\Delta))$  on the causal base $\Sigma'$. {Here $J^\pm(\Delta):=\{\mathfrak{x}+\mathfrak{z}: \mathfrak{x} \in\Delta, \mathfrak{z}  \text{ nonspacelike } ,  \pm z_0\ge 0 \} $.}}
	\label{fig:region-of-influence}
\end{figure}
An AL $T$ satisfies
\begin{equation*}
T(\Delta)\le T(\Delta')\:.  \tag{CC}
\end{equation*}

{For the proof and more details see \cite[sec.\,4, in particular (16) Theorem]{C24}.}  Obviously $T^e$  satisfies CC, too. Note that CC implies the familiar causal time evolution, i.e., $W^e(t)T^e(\Delta)W^e(t)^{-1}\le T^e(\Delta_t)$ for and every region $\Delta$ of Euclidean space and the region of influence 
$\Delta_t:=\{y\in\R^3:\,\exists \,x\in\Delta  \text{ with } |y-x|\le |t|\}$,  $t\in\R$.

\begin{Cor}\label{ASRCL} Let $\Delta, \,\Gamma \in\mathcal{B}^{ach}$. If $\Delta\cup\Gamma$ is acronal then $T^\textsc{d}(\Delta)$ and $T^\textsc{d}(\Gamma)$ commute. If  $\Delta, \,\Gamma$  are achronally separated, then $T^\textsc{d}(\Delta)$ and $T^\textsc{d}(\Gamma)$ are orthogonal, i.e.  $T^\textsc{d}(\Delta)T^\textsc{d}(\Gamma)=0$.
\end{Cor}\\
{\itshape Proof.} Let $\Lambda\subset \R^4$ be maximal achronal containing $\Delta\cup\Gamma$. Then the claim follows from the fact that $T^\textsc{d}$ is a projection valued measure on $\Lambda$.\qed

As spacelike separateness implies achronal separateness it follows from (\ref{ASRCL}) that the AL $T^\textsc{d}$  and the RCL $F^\textsc{d}$ satisfy  {\bf local orthogonality} and  in particular local commutativity or  {\bf microscopic causality}. 
 \\  
\hspace*{6mm}
 However $T^e$ is no longer microscopically causal, i.e., $T^e(\Delta)$  and  $T^e(\Gamma)$ do not commute for  spacelike separated regions $\Delta$ and $\Gamma$. By
\cite[Theorem 2]{HC01} this holds true  for some subsets of parallel spacelike hyperplanes. The conclusion drawn by \cite{HC01}  in the comment to this theorem is that the (positive energy) Dirac theory permits superluminal signalling. We will get back to this subject in sec.\,\ref{EPR}.

\subsection{Localization properties of $T^e$}\label{LPTE}

There are the  important results (\ref{SESA}),\,(\ref{SESB}), (\ref{LSD})  and  (\ref{POLDE}) about  the space $\mathcal{H}^e$ and the localization $T^e$ on the Euclidean space $\R^3\equiv \{0\}\times \R^3$. One is going to extend (\ref{SESB}) and   (\ref{POLDE})  to $T^e$ on general maximal achronal regions.

\begin{The}\label{SESA}\emph{\cite[Corollary 1.7]{T92}} Let $\psi\in\mathcal{H}^e\setminus \{0\}$.  
Then the support of $\psi$\footnote{Generally  $\operatorname{supp}(\psi)\equiv     \operatorname{supp}([\psi])      :=\bigcap_{\psi'\in[\psi]}\overline{\{\psi'\ne 0\}}$. 
  Note that there is $\psi'\in[\psi]$ satisfying $\operatorname{supp}([\psi])=\overline{\{\psi'\ne 0\}}$  \cite[18.1]{C17}.} is the entire Euclidean space.
   \end{The}

 \begin{Cor}\label{SESB} Let $\psi\in\mathcal{H}^e\setminus \{0\}$.
  Let $\Delta\subsetneq \R^3$  be closed. Then $T^e(\Delta)\psi \ne \psi$.
  \end{Cor}  
  
 Moreover, one has  
  
\begin{The} \emph{\cite{H85}}\label{LSD}  Let  $K >2m$. 
 The spatial probability in any state $\psi$ satisfies
\begin{equation*}
\langle \psi, T^e(\{x\in \R^3:\,|x|>r\})\psi\rangle \notin \mathcal{O}(\operatorname{e}^{-Kr}), \quad r\to \infty
\end{equation*}
 \end{The}
 \\
 So causal time evolution not only forbids localized states in non-essentially dense regions  but  requires also a limited exponential decay of the spatial probability. The limit is determined by   the Compton wavelength  $\lambda_C=\frac{\hbar}{mc}$. This is an interesting behavior of free relativistic systems.  But certainly it is not a paradox as sometimes 
it is referred to in the literature, neither does it mean that ``arbitrarily good localization'' \cite[sec.\,5]{H01} is impossible. Indeed, the localization  $T^e$  satisfies CC {\it and} is separated (\ref{ASSLA}) on every causal base (\ref{PCB}).

 By relativistic symmetry Euclidean space is not distinguished from any other spacelike hyperplane. Actually one recalls \cite[sec.\,3]{C24} that relativistic symmetry combined with causality necessitates to consider the more general   maximal achronal sets and their Borel subsets. 
  \\  
\hspace*{6mm}
 One keeps in mind the fact (\ref{AADAC}) that with respect to achronal separateness the completion $\Delta^\land$ of an achronal set $\Delta$  coincides with  its set  $\Delta^\sim$ of determinacy.
We will use this characterization of the completion as it is natural regarding the statements and proofs in the following. Recall that $\Delta$ is said to be determining if 
$\Delta^\sim\setminus\Delta \ne \emptyset$. It means that there is  $\mathfrak{x} \not\in \Delta$ such that every timelike  straight line through $\mathfrak{x}$  meets $\Delta$.
 \\  
\hspace*{6mm}
The result (\ref{SESB}) is a special case of  (\ref{ESES}).

\begin{The}\label{ESES} Let $\Delta\in\mathcal{B}^{ach}$. Suppose there is a determining $\Gamma\in\mathcal{B}^{ach}$ achronally separated from $\Delta$.  Then $T^e(\Delta)\psi \ne \psi$ for  $\psi\in\mathcal{H}^e\setminus \{0\}$.
 \end{The} 
 
{\it Proof.}  We anticipate (\ref{SPNO}). Accordingly, $\Gamma^\sim$      contains  $\mathfrak{a}+B_R$ for some    $\mathfrak{a}$ and $R>0$ with $B_R:=\{\mathfrak{y}: y_0=0, |y|<R\}$. By translation-covariance it is no restriction to assume     $\mathfrak{a}=0$.              
Note $\Gamma\subset \Delta^\perp$, whence $B_R\subset \Gamma^\sim = \Gamma^\land \subset \Delta^\perp$.
Hence $\Delta^\land\subset B_R^\perp =\{\mathfrak{x}:|x_0|\le |x|-R\}$. It follows $T^e(\Delta)=F^e(\Delta^\land)\le F^e(B_R^\perp)=T^e(\{|x|\ge R\})$ as $\{\mathfrak{x}: x_0=0,|x|\ge R\}$ is maximal achronal in $B_R^\perp$. By (\ref{SESB}) one has $T^e(\{|x|\ge R\})\psi\ne\psi$ for $\psi\in\mathcal{H}^e\setminus \{0\}$, whence the result.\qed
  
 \begin{Cor}  Let $\Delta\in\mathcal{B}^{ach}$ be bounded. Then $T^e(\Delta)\psi \ne \psi$ for  $\psi\in\mathcal{H}^e\setminus \{0\}$.
 \end{Cor} 
 
 {\it Proof.} Choose $\mathfrak{a}\in \R^4$ such that $\Gamma:=\mathfrak{a} +\{\mathfrak{y}: y_0=0, |y|<1\}$ is achronally separated from $\Delta$. Apply (\ref{ESES}).\qed
  
A state $\psi\in \mathcal{H}^e$, $||\psi||=1$ is said to be {\bf localized} in $\Delta \in \mathcal{B}^{ach}$ if $T^e(\Delta)\psi=\psi$. Hence by (\ref{SESB}),  there are no localized electron states in proper closed subsets and a fortiori in bounded regions of Euclidean space.
(The question is left open whether there is a state localized in an open dense $\Delta$ such that $\R^3\setminus\Delta$ has positive Lebesgue measure.) Equivalently the  probability of localization of the electron  in a  proper closed subset of Euclidean space is less than $1$.  
\\
\hspace*{6mm}
So   a more general property of localization   is introduced \cite{C81}, which is mathematically weaker than that of a localized state  in $\Delta$ but physically equivalent to it, i.e., $||T^e(\Delta)|| = 1$.  Indeed, as
\begin{equation*} \label{CID}
||T^e(\Delta)||= \sup \left\{\langle \psi, T^e(\Delta)\psi\rangle:\, ||\psi||=1\right\}
\end{equation*}
 holds, norm $1$  means that  the electron can be localized within that region  $\Delta$ by a suitable preparation, not strictly but as accurately as desired. By a consequence of \cite[Theorem 8]{CL15} one has the physically very satisfactory property regarding $T^e$  on Euclidean space.

\begin{The}\label{POLDE} \emph{\cite[(88)]{C17}}  Let the Borel set $\Delta$ of Euclidean space have an interior point. Then $||T^e(\Delta)||=1$.
\end{The}

The following statement  is obviously equivalent to  $||T^e(\Delta)||=1$ for all  Borel $\Delta$ with interior points:
For  every point $a$ of Euclidean space
there is a sequence $(\psi_n)$ of electron states satisfying 
\begin{equation}\label{ASSL}
\big\langle \psi_n,T^e(B )\, \psi_n\big\rangle \to 1,  \quad n\to \infty
\end{equation}
for every  open ball $B$ around $a$.  According to  \cite[sec.\,G]{CL15}, $(\psi_n)$  is called a sequence of states localized  at $a$,  and  $T$ is called  separated.  Point-localized sequences of states are studied and constructed  in \cite[sec.\,H]{CL15}.    See also   \cite[sec.\,8,\,18]{C17}. Explicit examples of point-localized sequences of states regarding the electron localization $T^e$ are given in \cite[sec.\,I]{CL15} and  \cite[sec.\,18.2]{C17}. Certain of those 
  have already been  constructed and studied in great detail  in   \cite{BM99} and     \cite{BFM05}. The concept there is further formalized in   \cite{M02}. In regard to \cite[(92) Example]{C17} see also \cite[sec.\,2]{BK03}.  The research on the subject still goes on as in \cite{BT26}.
 \\
\hspace*{6mm}
 According to (\ref{POLDE}), $T^e$ on Euclidean space  is separated. Hence by a suitable preparation the electron can be localized around $a$ as good as desired, thus distinguishing $a$ from any other point.

Regarding general achronal sets  $\Delta\subset \R^4$,
by definiton $\mathfrak{x}\in\Delta$ is an interior point of $\Delta$ if there is an open ball $B$ in $\R^3$ around $x$ contained in $\varpi(\Delta)$. Now
let $\Lambda$ be any maximal achronal set with $\Delta\subset \Lambda$. Consider $\Lambda$ to be a subspace of $\R^4$ equipped with the trace topology. Let  $\mathfrak{x}\in\Lambda$. The open {\bf achronal balls}  $U:= \Lambda\cap (\R\times B)$ around $\mathfrak{x}$, where $B$ is an open ball in $\R^3$ around $x$, constitute a local base of $\mathfrak{x}$. Obviously $\varpi(U)=B$.
So $\mathfrak{x}$ is an interior point of $\Delta$ if and only if $\Delta$ is a neighborhood of $\mathfrak{x}$, i.e., $\Delta$ contains an open achronal ball around  
 $\mathfrak{x}$. 
 \\
 \hspace*{6mm}
 Let $\psi_n\in\mathcal{H}^e$, $||\psi_n||=1$, and $\mathfrak{x}$ an interior point of $\Delta$. By definition  $(\psi_n)$ is a {\bf sequence of states localized at} $\mathfrak{x}$  if
 \begin{equation}\label{ASSLA}
\big\langle \psi_n,T^e(U )\, \psi_n\big\rangle \to 1,  \quad n\to \infty
\end{equation} 
for every  open achronal ball $U$ around $\mathfrak{x}$ contained in $\Delta$. If $\Delta$ is maximal achronal and if there is a point-localized sequence for every $\mathfrak{x}\in\Delta$, then $T^e$ is said to be {\bf separated on} $\Delta$.
 \\
 \hspace*{6mm}
The following result  (\ref{SPNO}) is decisive. On it depend (\ref{ESES}), (\ref{CSPNO}) - (\ref{PCB}). 

\begin{The}\label{SPNO} Let $\Delta\in\mathcal{B}^{ach}$ be
determining. Then  $||T^e(\Delta)||=1$. Let $(\alpha,a)\in \Delta^\sim\setminus \Delta$. Put $\delta:=|\alpha-\beta|$ for  $(\beta,a)\in\Delta$. 
 Then $$(\text{\SMALL{$\frac{\delta}{2}$}},a)+  \{\mathfrak{y}=(0,y):|y|<  \text{\SMALL{$\frac{\delta}{2}$}} \}\subset \Delta^\sim$$
and  $(\beta,a)$  is an interior point of $\Delta$.
\end{The}

{\it Proof.} Let $\Delta=\operatorname{graph}\tau$ with $\tau:\varpi(\Delta)\to \R$ a $1$-Lipschitz function.   Fix $\mathfrak{a}\in \Delta^\sim\setminus \Delta$.  By translation covariance it suffices   to treat the case  $\mathfrak{a}=(\alpha,0)$ and   $0\in \Delta$. Assume $\alpha>0$ ($\alpha<0$ is treated analogously). Then 
$|\tau(x)|\le  |x|$ for $\tau(0)=0$.

(a) As $\mathfrak{a}\in\Delta^\sim$,  for every $z\in\R^3$ with $|z|<1$ there is just one $t=T(z)\in\R$ such that $\mathfrak{a}+t(-1,z)=(\alpha-t,tz) \in\Delta$. 
\begin{itemize}
\item[(i)] $T(z)>\frac{\alpha}{2}$ for $|z|<1$
\end{itemize}
Indeed, note $ \alpha-T(z)=\tau\big(T(z)z\big)$ and   $|\tau\big(T(z)z\big)|\le |T(z)z|<|T(z)|$, whence the claim
\begin{itemize}
\item[(ii)] $T$ is continuous.
\end{itemize}
 Indeed, $|T(z')-T(z)|  = |  \tau\big(T(z')z'\big)-\tau\big(T(z)z\big)|  \le | T(z')z'- T(z)z|  =| \big(T(z')-T(z)\big)z'+ T(z)(z'-z)|)\le |T(z')-T(z)||z'|+|T(z)||z'-z| $, whence 
$|T(z')-T(z)|\le T(z)\frac{|z'-z|}{|1-|z'||}\to 0$ for $z'\to z$.

\hspace*{6mm}
  Let $\mathfrak{y}:=(\frac{\alpha}{2},y)$ with $|y|<\frac{\alpha}{2}$. We are going to show that $\mathfrak{y}\in \Delta^\sim$, i.e., that every straight timelike line $\mathfrak{y}+s(-1,e)$,  $|e|<1$, $s\in\R$  meets $\Delta$.

(b) Note that   $\mathfrak{y}+s(-1,e) \in\Delta$  implies $s>0$. Indeed,  as $\tau(y+se)=\frac{\alpha}{2}-s$, one has $|\frac{\alpha}{2}-s|\le |y+se|\le |y|+|s|<\frac{\alpha}{2}+|s|$, whence the claim.  Moreover, if $y=\frac{\alpha}{2}e$, then $\mathfrak{y}+\R(-1,e)$ insersect $\Delta$ at  $\mathfrak{a}+T(e)(-1,e)$. Therefore one may assume $s>0$ and  $y\ne\frac{\alpha}{2}e$. The following is easy to verify. 
 \begin{itemize}
\item[(iii)]   Let $y\ne\frac{\alpha}{2}e$. Then $\mathfrak{a}+\R(-1,z)$ and  $\mathfrak{y}+\R_+(-1,e) $  intersect at  $\mathfrak{a}+t(-1,z)$ if and only if $$t>\frac{\alpha}{2} \text{ \,and\, \,} z= e +\frac{1}{t}(y-\frac{\alpha}{2}e)$$
 It follows $z\ne e$ and $t=\frac{|y-\frac{\alpha}{2}e|}{|z-e|}$.
\end{itemize}

(c) Put $Z:[\frac{\alpha}{2},\infty]\to \{z:|z|<1\}$, $Z(t):=e +\frac{1}{t}(y-\frac{\alpha}{2}e)$ for $t<\infty$ and $Z(\infty):=e$. Obviously $Z$ is continuous. By (ii), $X:[\frac{\alpha}{2},\infty[\to \R$, $X(t):=T\big(Z(t)\big) -    \frac{|y-\frac{\alpha}{2}e|}{|Z(t)-e|}$ is continuous. Note that $T\big(Z([\frac{\alpha}{2},\infty])\big)$ is compact and, by (i), contained in $]\frac{\alpha}{2},\infty[$,  $X(\frac{\alpha}{2})= T(\frac{2}{\alpha}y)-\frac{\alpha}{2}>0$, and $X(t)\to -\infty$ for $t\to\infty$. Hence by the  intermediate value theorem there is $t^*>\frac{\alpha}{2}$ with $X(t^*)=0$. It follows $t^*= \frac{|y-\frac{\alpha}{2}e|}{|Z(t^*)-e|} = T\big(Z(t^*)\big) $. So, by (iii),  $\mathfrak{y}+\R(-1,e) $ intersects $\Delta$ at $\mathfrak{a}+t^*(-1,Z(t^*))$.
  
(d) Hence $K:=( \frac{\alpha}{2},0)+  \{\mathfrak{y}=(0,y):|y|<\frac{\alpha}{2}\}\subset \Delta^\sim$. $||T^e(K)||=1$ by (\ref{POLDE}) due to time-translation covariance of $T^e$.  Recall  $\Delta^\sim=\Delta^\land$ (\ref{AADAC}) and $F^e(A^\land)=T^e(A)$ (\ref{MTCL}) for $A =K,\,\Delta$. It follows $K^\land\subset \Delta^\land$ and $T^e(K)=F^e(K^\land) \le F^e(\Delta^\land)=T^e(\Delta)$, whence  $||T^e(\Delta)||=1$. Finally note $\{ y\in\R^3: |y|<\frac{\alpha}{2}\}=\varpi(K)\subset \varpi(\Delta^\sim) = \varpi(\Delta)$ thus finishing the proof.
\qed

\begin{The}\label{CSPNO} Let $\Delta \in\mathcal{B}^{ach}$ be compact and spacelike. Suppose $(\alpha,a)\in \Delta^\sim\setminus \Delta$.  
Put $\delta:=|\alpha-\beta|$ for  $(\beta,a)\in\Delta$. Let $0\le \lambda_k\uparrow_k 1$. Let $(\psi^a_n)$  be a sequence of states localized  at $a$ in Euclidean space. Given $k\in\N$ choose  a subsequence $n_k\in\N$ such that 
$$\langle \psi^a_{n_k}, T^e(\{y\in\R^3:|y|< 2^{-k}\delta \})  \psi^a_{n_k} \rangle\ge \lambda_k$$
Define  by  time-translation $$\psi_k:=W^e\big(2^{-k}\delta\big) \psi^a_{n_k}$$ 
Then  $(\beta,a)$ is an   interior point of $\Delta$ by \emph{(\ref{SPNO})} and  $(\psi_k)$ is a sequence of states localized at $(\beta,a)$.
\end{The}

{\it Proof.}  Let $\Delta=\operatorname{graph}\tau$ with $\tau:\varpi(\Delta)\to \R$ a $1$-Lipschitz function.  By translation covariance it suffices   to treat the case  $\mathfrak{a}=(\alpha,0)$ and   $0\in \Delta$. Assume $\alpha>0$ ($\alpha<0$ is treated analogously). Then $(\beta,a)=0$, $\delta=\alpha$, and $|\tau(x)|<  |x|$ for  $x\ne 0$ since $\tau(0)=0$ and $\Delta$ spacelike.

(a) 
By time-translation covariance $$\langle \psi_k, T^e(K_k)  \psi_k \rangle\ge \lambda_k, \quad 
K_k:=(2^{-k}\alpha,0)+ \{(0,y)\in\R^4:|y|< 2^{-k}\alpha\}$$
 For $k=1$, $K_1$ equals $K$ in (d) of the proof of (\ref{SPNO}). 
Hence $K_1\subset \Delta^\sim$, whence 
$(2^{-1}\alpha,0)\in\Delta^\sim\setminus \Delta$. Now one repeats the construction in (\ref{SPNO}) for $(2^{-1}\alpha,0)$ in place of $(\alpha,0)$ and so on. Put 
$$\Delta_k:=\{\mathfrak{x}=(2^{-k+1}\alpha-T_k(z),\, T_k(z)z):|z|<1\}$$
where,  $T_1:=T $  refers to $(\alpha,0)$, and analoguously $T_k$ refers to $(2^{-k}\alpha,0)$. Obviously $0\in \Delta_k\subset \Delta$ and $K_k\subset \Delta_k^\sim$. Hence $T^e(K_k)\le T^e(\Delta_k)$, see (d) of the proof of (\ref{SPNO}).

(b) The claim is $\Delta_k\subset \Delta_l$ if $l<k$. Indeed, let $\mathfrak{x}=(2^{-k+1}\alpha-T_k(z),\, T_k(z)z) \in\Delta_k$. Note  $\alpha_k<\alpha_l$. Put $t:=\alpha_l-\alpha_k+T_k(z)$ and $z'=T_k(z)/t$. Then $|z'|<1$ and $\mathfrak{x}=(2^{-l+1}\alpha- t,\, tz')$. This proves $\mathfrak{x}\in\Delta_l$ with  $t=T_l(z')$, see (a) of the proof of (\ref{SPNO}).

(c) Let $R>0$. The claim is $\varpi(\Delta_k) \subset B_R:=\{x\in\R^3:|x|<R\}$ for some $k$. 
Indeed, assume the contrary, i.e., for every $k$ there is $x_k\in\varpi(\Delta_k)$ with $|x_k|\ge R$.
As $\varpi(\Delta)$ is compact and $(x_k)$ in   $\varpi(\Delta)$, there is $(x_{k_l})$ and $x\in \varpi(\Delta)$ with $x_{k_l}\to x$. 
\\
\hspace*{6mm}
Then $|x|\ge R$, whence $x\ne 0$.  As $x_k\in\varpi(\Delta_k)$, $(\tau(x_k),x_k)\in\Delta_k$ and hence $(\tau(x_k),x_k)=(2^{-k+1}\alpha-T_k(z_k), T_k(z_k)z_k)$ for some $|z_k|<1$. Assume without restriction $z_{k_l}\to z$, $|z|\le1$. Put $t_l:=T_{k_l}(z_{k_l})$. It follows $\lim_lt_l=-\tau(x)$, $\lim_l t_lz_{x_l} =-\tau(x) z=x$. As $|\tau(x)|<|x|$ one has  the contradiction $|\tau(x)|< |\tau(x)||z|\le|\tau(x)|$.

(d) Let $U\subset \Delta$ be an open achronal ball around $0$.  Let $\epsilon>0$. By (c) and (b)  there is $l$ such that $\Delta_k\subset U$ and  
$\langle \psi_k, T^e(K_k)  \psi_k \rangle\ge 1-\epsilon$ for all $k\ge l$. Hence, by (a),  $\langle \psi_k, T^e(U)  \psi_k \rangle\ge 1-\epsilon$ for all $k\ge l$.\qed

We like to recall the comments on (\ref{ASSL}), whence  explicit examples for point-localized sequences of states $(\psi_k)$ in (\ref{CSPNO}) are available.

\begin{The}\label{SLIP}  Let $\Delta\in \mathcal{B}^{ach}$ be spacelike and have an interior point. Then 
 $\Delta$ is determining and, by \emph{(\ref{SPNO})},  $||T^e(\Delta)||=1$.
\end{The}

{\it Proof.} Without restriction let the interior point of $\Delta$ be the origin $0$ and let  $\varpi(\Delta)=B_R$, where $B_R:=\{x\in\R^3: |x|\le R\}$ for some finite $R>0$. Let $\Delta=\operatorname{graph}(\tau)$. One has $|\tau(x)|<|x|\le R$ as $\tau(0)=0$. Since $\tau(B_R)$ is compact there is $0<\delta< R$ such that $|\tau(x)|\le R-\delta$, whence in particular $\tau(x)\ge \delta -R$ for $|x|\le R$.
\\
\hspace*{6mm}
One is going to show that every causal straight  half-line $$\gamma(t):= (\delta-t,tz), \;t\ge 0$$ for $z\in\R^3$, $|z|\le1$ meets $\Delta$, whence $(\delta,0)\in \Delta^\sim\setminus \Delta$ proving the claim.  
\\
\hspace*{6mm}
Put $\Delta^+:=\{\mathfrak{x}:x_0\ge \tau(x)\}$ and similar $\Delta^-$. Note $\gamma(0)\in\Delta^+$ for every $z$ as $\delta>0$. For $z=0$, $\gamma(R)\in\Delta^-$ as $\delta<R$. Now let $z\ne0$. Then $\gamma(\frac{R}{|z|}) \in\Delta^-$, since  $\delta -\frac{R}{|z|}\le \delta- R\le \tau(\frac{R}{|z|}z)$. It follows by the intermediate value theorem that $\gamma(t^*)\in\Delta$ for some $t^*>0$.\qed

\begin{Cor}\label{PCB}
Let $\Sigma$ be a causal base. Recall that  $\Sigma$  is maximal achronal, whence every point of $\Sigma$ is an interior point of $\Sigma$.  There is a  point-localized sequence of states for every point in  $\Sigma$, i.e., $T^e$ is separated on $\Sigma$. Let $U \subset \Sigma$ be an  open achronal ball. Then in particular  $||T^e(U)||=1$. Actually 
$U$ is dertermining.
\end{Cor}
\\
{\it Proof.}  Let  $\mathfrak{a}   \in \R^4\setminus \Sigma$.  Again one may assume $0\in\Sigma$ and $\mathfrak{a}=(\alpha,0)$ with $\alpha>0$.
 Let $C_\alpha$ be  the past  light cone with vertex $\mathfrak{a}$.  Then every straight half-line in $C_\alpha$ starting at $\mathfrak{a}$ meets $\Sigma$. As shown in the proof of  \cite[(36)]{C24} $\Delta_\alpha:=\Sigma \cap C_\alpha$ is compact.
Then the result except for the last claim holds by (\ref{CSPNO}) and (\ref{ASSLA}). The last claim holds by (\ref{SLIP}).\qed

For the result (\ref{PCB}) one recognizes that regarding achronal  localization causal bases generalize adequately the spacelike hyperplanes regarding a localization for flat spacelike regions.
\\
\hspace*{6mm}
 As revealed in sec.\,\ref{HBLALLA} the light hyperplanes play a decisive role in achronal localzation due to their particular features. A  peculiarity is the fact that any bounded subset $\Delta$ of  a light hyperplane is already complete, i.e., it satisfies $\Delta=\Delta^\sim=\Delta^\land$ even if $\Delta$ has an interior point. Hence (\ref{SPNO}) does not apply
and it is left open whether $||T^e(\Delta)||=1$ holds true for an open bounded subset $\Delta$ of a light hyperplane.
\\
\hspace*{6mm}
 Let $\kappa$ be a light hyperplane. This means $\kappa=\{\mathfrak{x}\in\R^4: \mathfrak{x}\cdot \mathfrak{e}=\rho\}$ for some unique $\rho\in\R$ and  
    lightlike  $\mathfrak{e}=(1,e)\in\R^4$, $|e|=1$.  Let $\mathfrak{a}\in\kappa$ and $\lambda\in\R \setminus \{0\}$. The  subset 
\begin{equation}\label{OAP}
\mathfrak{a}+\{\mathfrak{x}\in\R^4: \mathfrak{x}\cdot \mathfrak{e}=0, |x|^2-(xe)^2< 2\lambda\, xe\} 
\end{equation}
of $\kappa$ is called an open achronal {\bf paraboloid}.

\begin{Pro}\label{COAP}      Let $\Delta\subset \kappa$. Then 
$\Delta$ is determining if and only if $\Delta$ contains an open achronal paraboloid.
\end{Pro}    
\\
{\it Proof.}      Let $\mathfrak{y}\not\in \kappa$. We  determine the intersection $S$ of $\kappa$ with the interior of the light-cone with apex $\mathfrak{y}$.  Note  $\mathfrak{z}\cdot \mathfrak{e}\ne 0$ for $\mathfrak{z}^{\cdot 2}>0$. So for every timelike $\mathfrak{z}$ one finds  $(\mathfrak{y}-s\mathfrak{z})\cdot \mathfrak{e}=\rho$ with $s=\lambda/(\mathfrak{z}\cdot \mathfrak{e})$ for $\lambda:=\rho-\mathfrak{y}\cdot\mathfrak{e}\ne 0$. One  checks that $\mathfrak{x}=\mathfrak{z}/(\mathfrak{z}\cdot \mathfrak{e})$ for some timelike $\mathfrak{z}$ if and only if $\mathfrak{x}\cdot \mathfrak{e}=1$ and $|x|<|x_0|$. Hence $S=\mathfrak{y}+\lambda\{\mathfrak{x}:\mathfrak{x}\cdot \mathfrak{e}=1, |x|<|x_0|\}$.  By the substitution $\mathfrak{x}=\frac{1}{\lambda}\mathfrak{x}'+\frac{1}{2}(1,-e)$,  $S$ equals  (\ref{OAP}) 
for  $\mathfrak{a}:=\mathfrak{y}+\frac{1}{2} \lambda \,(1,-e) \in\kappa$.
\\
\hspace*{6mm}
Hence   $\mathfrak{y}\in \Delta^\sim\setminus \Delta$ if and only if $S\subset\Delta$.
Conversely, let $Q$ be an achronal paraboloid be contained in $\Delta$. Put  
$\mathfrak{y}:= \mathfrak{a}-\frac{1}{2} \lambda \,(1,-e)\not\in\kappa$ as $\lambda\ne 0$. Then $Q$ equals $S$. This ends the proof.
\qed

\begin{Cor} Let $\Delta\in\mathcal{B}^{ach}$ be a subset of a light hyperplane  $\kappa$. Let $\Delta$ contain an open achronal paraboloid.  Then $||T^e(\Delta)||=1$ and no state $\psi\in\mathcal{H}^e$ is localized in $\kappa\setminus \Delta$.
\end{Cor}
\\
{\it Proof.}   Combine (\ref{SPNO}),\,(\ref{ESES}),\,(\ref{COAP}).\qed

\subsection{High boost limit, achronal localizability, Lorentz contraction}\label{HBLALLA}
In the limit of infinite rapidity a  spacelike hyperplane,  boosted along a direction parallel to it, equals   a {\bf light hyperplane}, i.e. an achronal  hyperplane, which is {\it not  spacelike}.
 Light hyperplanes are the tangent spaces of the light cones and, like the light cones,  are constituted by lightlike straight lines. 
The limit is called \textbf{high boost limit} if it occurs pointwisely   such that every point of spacetime runs along a lightlike straight line. In this case for a causal localization   the probabilities of localization converge by  continuity, essentially due to  monotony. This is investigated in great detail for the massive scalar boson in \cite{C25}, and for the Dirac system, the Dirac fermions, the Weyl systems and the Weyl fermions in \cite[sec.\,26, in particular (133)]{C17}.

\textbf{Achronal localizability.} So roughly speaking  by {\it continuity} every causal localization on flat spacelike regions automatically comprises the regions of all achronal hyperplanes.   It is worth to point out that the set of spacelike hyperplanes and the  disjoint set of light hyperplanes are not isolated although both  are  Poincar\'e invariant.
\\
\hspace*{6mm}
The existence of the high boost limit is decisive as in the end the conclusion to draw from this particular property of causal localization is  that localizability of a relativistic quantum mechanical system concerns not only flat spacelike regions, but all  {\em achronal} regions of spacetime \cite[sec.\,3]{C24}. Poincar\'e covariant AL and equivalently RCL constitute  the frame which complies most completely with this basic insight.
\\
\hspace*{6mm}
We proceed  providing a  proof of the above mentioned continuity for $T^\textsc{d}$ and  $T^e$.  In this case  it is easy to do due to the integral representation (\ref{REXAL}). A consequence of this continuity is the Lorentz contraction  for the electron.

\textbf{High boost limit of a spacelike hyperplane. }The light hyperplane $\chi:=\{\mathfrak{x}\in\R^4: x_0=x_3\}$ is the high boost limit of the Euclidean space $\varepsilon=\{\mathfrak{x}\in\R^4: x_0=0\}$ as follows.  In the same way, by relativistic symmetry, every light hyperplane is the high boost limit of a spacelike hyperplane. 
 Note $$A_\rho\cdot \varepsilon=\{x_0=\tanh(\rho)\, x_3\}$$ where  $A_\rho=\e^{\rho\, \sigma_3/2}$ represents the boost along the third spatial axis with rapidity $\rho$. 
 For $\rho\ge 0$ let
$$l_\rho:\varepsilon\to A_\rho\cdot \varepsilon, \quad l_\rho(0,x):=\big( \text{\footnotesize{$\frac{1}{2}$}}    (1-\e^{-2\rho})x_3,\,x_1,\, x_2,\, \text{\footnotesize{$\frac{1}{2}$}}    (1+\e^{-2\rho})x_3\big)$$ 
which is a linear  bijection composed by the inhomogeneous dilation $\mathfrak{x}\mapsto (x_0,x_1,x_2,\e^{-\rho}x_3)$ and the subsequent boost $\mathfrak{x}\mapsto A_\rho\cdot \mathfrak{x}$. One has  $ l_\rho \to l_\infty$ pointwisely for $\rho\to\infty$, where $l_\infty$ is a linear bijection given by
$$   l_\infty:\varepsilon \to \chi, \quad        l_\infty (0,x):=(\text{\footnotesize{$\frac{1}{2}$}} x_3,x_1,x_2,\text{\footnotesize{$\frac{1}{2}$}} x_3)$$
The maps $l_\rho$ are charaterized by fact that  every point $(0,x)\in \varepsilon$ runs through the segment $\{l_\rho(0,x): 0\le \rho\le \infty\}$ of the lightlike line $(0,x)+\R(\frac{x_3}{2},0,0,-\frac{x_3}{2})$ joining $(0,x)$ with $(\frac{x_3}{2},x_1,x_2,\frac{x_3}{2})\in\chi$.

\textbf{Continuity of} $T^\textsc{d}$.
The localization operators  on regions of $\varepsilon$ and  $\chi$ are related to each other by the high boost limit as follows.

\begin{The}\label{CTHBL}
Let $\psi$ be a state of the Dirac system and $\Delta$ a  Borel subset of $\varepsilon$. Then  $$\lim_{\rho\to\infty} \langle \psi,T^\textsc{d}\big(l_\rho(\Delta)\big)\psi\rangle=
\langle \psi,T^\textsc{d}\big(l_\infty(\Delta)\big)\psi \rangle$$
\end{The}\\
{\it Proof.} Put $t_\rho:=\tanh \rho$. Without restriction let $\psi\in\mathcal{D}$. Then by (\ref{REXAL})
\begin{equation*}\label{IREXAL}
\langle \psi, T^\textsc{d}(l_\rho(\Delta))\psi\rangle= \int_{\varpi(l_\rho(\Delta))} \big( J_0(\psi;t_\rho x_3,x) - t_\rho J_3(\psi;t_\rho x_3,x)\big)  \d^3x 
\end{equation*}
Obviously $1_{\varpi(l_\rho(\Delta))}\to 1_{\varpi(l_\infty(\Delta))}$ and $\big(J_0(\psi;t_\rho x_3,x) - t_\rho J_3(\psi; t_\rho x_3,x)\big) \to \big( J_0(\psi;x_3,x) -  J_3(\psi;x_3,x)\big)$
 pointwisely for $\rho\to\infty$. By (\ref{PDPC})(d) and (\ref{NER}) one gets the integrable majorant $2C_4(1+|x|)^{-4}$ of the integrand, whence the claim by dominated convergence and (\ref{REXAL}).\qed

\textbf{Lorentz contraction}, see \cite{C17} and cf.\,\cite{C25}. Recall that $A_{\rho e}:=\operatorname{exp}(\frac{\rho}{2}\sum_{k=1}^3e_k\sigma_k)$  represents the boost in direction $e\in\R^3$, $|e|=1$ with rapidity $\rho$. 

\begin{The}\label{LCMSB} Let $\psi$, $||\psi||=1$, be a state of the Dirac system.  Then
$$\langle\, W(A_{\rho e})\psi, T^\textsc{d}(\{\mathfrak{x}\in\varepsilon:-\delta \le xe \le \delta\}) \, W(A_{\rho e})\psi\, \rangle \to 1,\quad |\rho|\to\infty$$
 for every   $\delta>0$.
\end{The}\\
{\it Proof.} It suffices to treat the case $\rho\to\infty$ and every $e$. Indeed, for the case $\rho\to -\infty$ consider $-e$. Then, due to Euclidean covariance,  it suffices to deal  with only one  direction $e$. We choose $e=(0,0,-1)$. 
\\
\hspace*{6mm}
Let $\epsilon>0$. There is $0<\beta<\infty$ such that $\langle \psi, T^\textsc{d}(\{\mathfrak{x}\in\chi:|x_3|\le \beta/2\}) \,\psi\rangle > 1-\epsilon$  since $T^\textsc{d}(\chi)=I$. Note $\{\mathfrak{x}\in\chi:|x_3|\le \beta/2\} =l_\infty(\{\mathfrak{x} \in\varepsilon:|x_3|\le \beta\})$. 
Let $\rho_\beta>0$ with $\beta\e^{-\rho}\le \delta$ for $\rho\ge \rho_\beta$. 
\\
\hspace*{6mm}
By (\ref{CTHBL}), 
$1\ge \langle W(A_{-\rho})\psi, T^\textsc{d}(\{\mathfrak{x}\in\varepsilon:|x_3| \le \delta \})\, W(A_{-\rho})\psi\rangle \ge \langle W(A_{-\rho})\psi, T^\textsc{d}(\{\mathfrak{x}\in\varepsilon:|x_3| \le \beta\e^{-\rho}\})\, W(A_{-\rho})\psi\rangle  = \langle \psi, T^\textsc{d}(A_\rho\cdot \{\mathfrak{x}\in\varepsilon:|x_3| \le \beta\e^{-\rho}\})\,\psi\rangle =\langle \psi, T^\textsc{d}(l_\rho(\{\mathfrak{x}\in\varepsilon:|x_3| \le \beta\}))\,\psi\rangle  \to_\rho \langle \psi, T^\textsc{d}(\{\mathfrak{x}\in\chi:|x_3|\le \beta/2\}) \,\psi \rangle > 1 - \epsilon $.
Hence $1\ge  \langle W(A_{-\rho})\psi, T^\textsc{d}(\{\mathfrak{x}\in\varepsilon:|x_3| \le \delta \})\, W(A_{-\rho})\psi\rangle \ge 1-\epsilon$ for all $\rho$  
large enough. The result follows.\qed
\\

\begin{Cor}\label{ELCMSB} The result $\emph{(\ref{LCMSB})}$ holds equally for the electron  AL $(W^e,T^e)$.
\end{Cor}
\\
{\it Proof.} Due to (\ref{EPL}) it suffices to choose $\psi\in\mathcal{H}^e$.\qed
\\

By  (\ref{ELCMSB}) the  probability of localization  of the electron in the boosted state $W(A_{\rho e})\psi$ in a whatever narrow strip $\{-\delta\le xe\le \delta\}$    tends to $1$ if the rapidity $\rho$ tends  to $\infty$ or $-\infty$. We like to call  this behavior  the \textbf{Lorentz contraction} of  the electron. 

Here the usual questions related to classical Lorentz contraction equally arise. The following discussion is taken from \cite{C25}. For more details cf. \cite[sec.\,17.2]{C17}.

\hspace*{6mm} 
(a) Can Lorentz contraction be observed?   Immagine an apparatus $\mathcal{A}$ able to ascertain the probabilities of localization of the electron in $\{|xe|\le \delta \}$, i.e.,  the expectation  values of $A:=T^e(\{|xe|\le \delta \})$. For a given state $S$ described by $\psi$ and $\varepsilon >0$, let $\delta>0$ be so small that $\langle \psi, T^e(\{|xe|\le \delta \})\psi\rangle \le \varepsilon$. According to (\ref{LCMSB}) there is a rapidity $\tilde{\rho}$ such that 
$\langle\, W(A_{\rho e})\psi, T^e(|xe| \le \delta\}) \, W(A_{\rho e})\psi\ \rangle \ge 1-\varepsilon$ for $\rho\ge \tilde{\rho}$. Let $\tilde{S}$ be the boosted state. It is described by  $\tilde{\psi}= W(A_{\tilde{\rho} e})\psi$. Then
\begin{equation}\label{LCDSM} 
\langle \psi,A \psi\rangle\le \varepsilon\;\textrm{ and }\;\langle \tilde{\psi},A \tilde{\psi}\rangle\ge 1-\varepsilon
\end{equation}
 Hence the apparatus $\mathcal{A}$ distinguishes the state $S$ from the boosted state $\tilde{S}$. So an observer  can ascertain the Lorentz contraction of the electron.
\\
\hspace*{6mm} 
(b) The ascertainments (\ref{LCDSM}) are related to some reference frame $\mathfrak{R}$. 
What are the ascertainments  of an observer related to any other frame $\mathfrak{R}'\equiv g^{-1}\cdot\mathfrak{R}$
 with $g\in \tilde{\mathcal{P}}$  provided with the  localization $(T^e)'$? Note that $A'=T^e(g\cdot\{x_0=0,|xe|\le\delta\})$ and  $\tilde{\psi}'=W(h')\psi'$ for $h'=ghg^{-1}$, $h:=A_{\tilde{\rho}e}$. For these well-known general relations see e.g. \cite[sec.\,VIII]{A69},  \cite[sec.\,17.2]{C17}. Then 
\begin{equation} 
 \langle \psi',A'\psi'\rangle\le\varepsilon \text{\, and } \langle \tilde{\psi}',A'\tilde{\psi}'\rangle\ge 1-\varepsilon
 \end{equation}
  Hence, observed from $\mathfrak{R}'$, the electron  in the state $S$ is highly localized in the spacelike region $g\cdot\{\mathfrak{x}: x_0=0, |xe|>\delta\}$, whereas in the boosted state $\tilde{S}$ it is highly localized in $g\cdot\{\mathfrak{x}:x_0=0, |xe|\le\delta\}$. The expected conclusion is that, due to relativistic symmetry,  the  Lorentz contraction of the massive scalar boson  can be ascertained  in the same way and with the same result by any Lorentz observer.
\\
\hspace*{6mm} 
 (c) On the other hand there is the dependence of the Lorentz contraction on the frame, which is discussed now.  
Due to the relativistic symmetry and the covariance of $T^e$ the result (\ref{LCMSB}) can be expressed equivalently in the following way. Let a four vector $\mathfrak{e}$ be called a spacelike direction if 
$\mathfrak{e}\cdot \mathfrak{e}=-1$.

\begin{Cor} \label{FDLCDL}  Let $\sigma$ be a spacelike hyperplane and $\mathfrak{e}$ a spacelike direction parallel to $\sigma$. Boost them along $\mathfrak{e}$ with rapidity $\rho$ obtaining 
$\sigma_\rho$ and $\mathfrak{e}_\rho$. Then
$$   \langle\psi,T^e\big(\{\mathfrak{x}\in\sigma_\rho: |(\mathfrak{x}-\mathfrak{o})\cdot \mathfrak{e}_\rho|\le \delta\}\big)\psi \rangle \to 1 \textrm{ for } |\rho|\to \infty$$ where  $\mathfrak{o}\in\sigma$ is the fixed point of the boost.
\end{Cor}

 Thus, if the frame is moving fast enough depending on the state, then  the electron is highly localized in a narrow strip perpendicular to the direction of  motion.
\\
\hspace*{6mm} 
 The dependence on the frame of the Lorentz contraction in classical  mechanics is striking by the fact that for the comoving observer it does not even exist. The same holds true for the Lorentz contraction of the electron wavefunctions. Moreover, due to  the Poincar\'e covariance of the localization no reference to a {\it moving observer} is needed,  but the fact refers to the expectation value of the corresponding localization observable. Indeed, boost the apparatus $\mathcal{A}$ according to  $A_{\tilde{\rho}e}$ thus obtaining the comoving apparatus $\tilde{\mathcal{A}}$. It is able to ascertain  the probabilities of localization of the electron in the spacelike region $A_{\rho e}\cdot\{|xe| \le \delta\}$ (to which a comoving observer refers), i.e., the expectation  values of  $\tilde{A}:=T^e\big(A_{\rho e}\cdot\{|xe| \le \delta\}\big)$. Then due to the covariance of localization
 \begin{equation} \label{NLCFCM}
\langle \tilde{\psi},A \tilde{\psi}\rangle\ge 1-\varepsilon\;\textrm{ and }\;  \langle \tilde{\psi},\tilde{A} \tilde{\psi}\rangle\le \varepsilon
\end{equation} 
holds. This means that the non-comoving apparatus $\mathcal{A}$ ascertains the Lorentz contraction of the electron whereas the comoving apparatus $\tilde{\mathcal{A}}$ ascertains non-contraction.

\subsection{Electron-positron representation}\label{EPR}
Let us point out that the Dirac system for every positive mass and the two right- and left-handed Weyl systems, which are the massless Dirac systems, are the only known irreducible {\it causal systems}. Here causal system means a rep of $\tilde{\mathcal{P}}$ and covariant {\it projection operator} valued map $T$ on the flat spacelike Borel sets, which is a normalized measure on every spacelike hyperplane and which satisfies CC. The  mentioned causal systems are the only ones with a nonnegative mass-squared operator,  a finite   spinor dimension, and which satisfy a  (presumably superfluous) regularity assumption.  For details see \cite[sec.\,27-29]{C17}.
\\
\hspace*{6mm} 
As stated in sec.\,\ref{CCMC}, $T^e$ is not microscopically causal. However,
despite this failure the physically relevant probabilities of localization of the electron are given by the expectation values of the localization operators  $T^\textsc{D}$ of the Dirac system (\ref{EPL}), which are  {\it locally orthogonal}. Let us try to  take seriously this fact. 
 The question is how to understand this peculiar situation which is unique as evidenced in the foregoing remark. 
 \\
\hspace*{6mm}
According to the result (\ref{SESB}) (see also the more general (\ref{ESES}) for spacetime regions) the electron cannot be localized in bounded regions of Euclidean space. The reason one imagines is that attempting the localization causes such a large  uncertainty in energy that
pair-production occurs.  In fact, we are going to argue that the  electron localization $T^e$ is in compliance with this explanation. It  takes account of pair-production in an effective manner within the frame of relativistic quantum mechanics.
 \\
\hspace*{6mm} 
  Here we present  the main idea and the outcomes (\ref{POE}),\,(\ref{POP}),  which concern the probability of the occurrence of an electron and a proton, and the final state of the Dirac system after the position measurement, the density operator (\ref{SO}). The results concern a {\em non-selective measurement}. They are {\em universally valid} in the sense that they hold true as long as no information gained from the actual measuring equipment is taken into account.
The subject is developed to some extent  in \cite[Sec.\,III\,J]{CL15} and \cite[sec.\,19]{C17}.
\\
\hspace*{6mm}
One starts recalling the fact  that the state space  of the Dirac system  $\mathcal{H}^\textsc{d}$ is the inner orthogonal sum 
\[
\mathcal H^\textsc{d}
=
\mathcal H^+\oplus\mathcal H^-
\]
 of the  $W^\textsc{d}$-invariant  irreducible eigenspaces $\mathcal{H}^\pm$  of  $\operatorname{sgn}(H)$ determined by the eigenvalues $\pm1$. Let the corresponding irreducible subrep of $W^\textsc{d}$ be denoted by $W^\pm$.  Here  $\mathcal{H}^+=\mathcal{H}^e$ and $W^+=W^e$, i.e.,  the {\it electron} representation characterized by mass the $m>0$, spin $j=\frac{1}{2}$ and the sign $+$ of energy.  The irreducible subrepresentation $W^-$ of $W^\textsc{d}$ on $\mathcal{H}^-$ is characterized by the same mass and spin but negative sign of energy. Hence,  $W^\pm\in[m,\frac{1}{2},\pm]$ for short.
\\
\hspace*{6mm}
In general for  spin $j\in\{0,\frac{1}{2}, 1,\dots\}$,  every   rep   in  $[m, j, -]$ is antiunitarily equivalent to every  rep  in $[m, j, +]$.
 The transforming operator is unique up to a constant phase. Also, if  $W$ is a rep in $[m, j, -]$ and if $\mathcal{A}$ is any antiunitary operator from the carrier space $\mathcal{H}'$ of $W'$ onto a Hilbert space $\mathcal{H}$, then  $W:=\mathcal{A}W'\mathcal{A}^{-1}$ is a rep in $[m, j, +]$ on $\mathcal{H}$. (Indeed, these facts  are an easy consequence of e.g. \cite[(93)]{C17}).
Let $H'$ be the energy operator, which is the generator of the time translations represented by $W'$. Then the self-adjoint operator 
$\mathcal{A}H'\mathcal{A}^{-1}$ is {\it not} the the energy operator $H$ referring to $W$. Indeed, $\langle \varphi, \mathcal{A}H'\mathcal{A}^{-1}\varphi\rangle = 
\langle H'\mathcal{A}^{-1}  \varphi,     \mathcal{A}^{-1}\varphi\rangle = \langle \mathcal{A}^{-1}  \varphi,   H'  \mathcal{A}^{-1}\varphi\rangle   \le 0$ for all states $ \varphi\in \mathcal{H}$. Hence $\mathcal{A}H'\mathcal{A}^{-1}$ is still negative, whereas $H\ge 0$. However, $W'$ represents the time translations by $W'(t)=e^{-\i t H'}$. As $W(t)=\mathcal{A}W'(t)\mathcal{A}^{-1}$ this implies $e^{-\i t H}=\mathcal{A} e^{-\i t H'} \mathcal{A}^{-1} =\exp(\mathcal{A}(-\i tH') \mathcal{A}^{-1}) =\exp\big(-\i t (-\mathcal{A}H'\mathcal{A}^{-1})\big)$. One infers 
\begin{equation}\label{EAP}
H= -\mathcal{A}H'\mathcal{A}^{-1}
\end{equation}
which is a nonnegative operator.
\\
\hspace*{6mm}
These facts suggest the following reasoning. We start out from the assumption that $W'\in [m, j, -]$ represents a stable particle, too. This has mass $m>0$, spin $j$ and of course {\it positive energy}.  Hence the sign $-$ does {\it not} refer to the energy but in actual fact indicates that the particle is the  antiparticle of the particle described by $W\in [m, j, +]$.
The representation of the  kinematical observables referring to the antiparticle  are gained from $W$ analogously to (\ref{EAP}). 
\\
\hspace*{6mm}
 Stating that  the antiparticle is in the state $\varphi\in \mathcal{H}'$  means that it is in the state $\mathcal{A}\varphi\in\mathcal{H}$. Below we tacitly use this meaning keeping in mind that  the transition probabilities $|\langle \varphi,\varphi'\rangle|^2 = |\langle \mathcal{A}\varphi,\mathcal{A}\varphi'\rangle|^2 $, in particular $||\varphi||= || \mathcal{A}\varphi||$, are preserved.
\\
\hspace*{6mm}
We turn to the case $W'=W^-$. For $\mathcal{A}$ one chooses the {\bf time reversal operator} $\mathcal{T}^-:=j^*_-\mathcal{T} j_-$ on $\mathcal{H}^-$, which is a symmetry 
operation, but any antiunitary operator $\mathcal{A}$ on $\mathcal{H}^-$ is just as good. Here $\mathcal{T}$ is the time reversal operator for the Dirac system and $j_-: \mathcal{H}^- \to  \mathcal{H}^\textsc{d}$ denotes  the identical injection.  In momentum space representation $\mathcal{T}$  reads 
\begin{equation}\label{TIO}
 (\mathcal{T}\varphi)(p)=    \left( \begin{array}{cc}  \i \sigma_2 & 0\\ 0 & \i \sigma_2 \end{array}\right)   \,\overline{\varphi(-p)} 
 \end{equation}
It is easy to verify that $\mathcal{T}$ and the Hamiltonian $H$ of the Dirac system commute $$[\mathcal{T},\, H] =0$$
In particular $[\mathcal{T},\, \operatorname{sgn}(H)] =0$ holds so that $\mathcal{T}$ preserves the subspaces $\mathcal{H}^\pm$.
\\
\hspace*{6mm}
According to the foregoing considerations $W^-\in [m,\frac{1}{2},-]$ indicates the presence of the antiparticle of the electron, the {\bf positron}, described by the rep $W^p\in [m,\frac{1}{2},+]$ defined as 
\begin{equation}
W^p:=\mathcal{T}^-  W^-     \,(\mathcal{T}^-)^{-1}
\end{equation}
Let  $\mathcal{H}^p$ denote the carrier space of $W^p$. Of course  $\mathcal{H}^p$ coincides with $\mathcal{H}^-$. 
\\
\hspace*{6mm}
We give an answer to  the question about the meaning of the states $\varphi$ of the Dirac system. It arises if neither  $\varphi\in \mathcal{H}^+$ nor $\varphi\in \mathcal{H}^-$, i.e., if $\varphi= \varphi_+ + \varphi_-$ is the linear combination of states $\hat{\varphi}_\pm:= ||\varphi_\pm||^{-1}\varphi_\pm\in\mathcal{H}^\pm$. In this case $\varphi$ is regarded to be  a {\it non-observable virtual superposition} of the electron state  $\hat{\varphi}_+$  and the positron state $\hat{\varphi}_-$. The attempt to realize this state by preparation creates a {\it mixed state} of the  electron state and  the positron state with respective probabilities $||\varphi_\pm||^2$. 
\\
\hspace*{6mm}
We are going to apply  this principle to our main subject, the electron localization. 
\\
\hspace*{6mm}
Let the electron be in the  state $\varphi_1\in  \mathcal{H}^e$, consider a region $\Delta\in\mathcal{B}^{ach}$, and suppose $T^e(\Delta)\varphi_1\ne 0$. 
Since $\varphi_1$ is a state of the Dirac system, $ \langle \varphi_1,T^\textsc{d}(\Delta) \varphi_1\rangle$ is the probability to get an affirmative answer at the position measurement to the question whether the Dirac system is in the region $\Delta$.  
Recall that
$\langle \varphi_1,T^e(\Delta)\varphi_1\rangle =\langle \varphi_1,T^\textsc{d}(\Delta) \varphi_1\rangle$. 
Hence 
$\langle \varphi_1,T^e(\Delta)\varphi_1\rangle$ is the probability  of localization in $\Delta$ of the electron in the state $\varphi_1$. It is also the probability  that after the measurement the Dirac system is in the state 
\begin{equation*}
\phi_1 := \norm{T^\textsc{d}(\Delta)\varphi_1 }^{-1} T^\textsc{d}(\Delta)\varphi_1
\end{equation*}
Now this state is a non-observable virtual superposition $$\phi_1=P^e\phi_1+ P^p\phi_1$$ of an electron and a positron state. Here $P^e$ is the projection operator on  $\mathcal{H}^e$ and $P^p := I-P^e$.
 Accordingly, after the measurement the electron state
\begin{equation*}
\varphi_2 :=   \norm{ P^e\phi_1 }^{-1}P^e\phi_1=\, \norm{ T^e(\Delta)\varphi_1 }^{-1} T^e(\Delta)\varphi_1
\end{equation*}
occurs with probability  $\sigma_2^e$ given by the product $\langle \varphi_1,T^e(\Delta)\varphi_1\rangle \norm{ P^e\phi_1}^2=\norm{ T^e(\Delta)\varphi_1 }^2$. 
\\
\hspace*{6mm} 
As $I-T^\textsc{d}(\Delta)= T^\textsc{d}(\Delta')$ with $\Delta' := \R^3\setminus \Delta$, one has at the same time the information that there is the electron state $\varphi_2':=  \norm{ T^e(\Delta')\varphi_1 }^{-1} T^e(\Delta')\varphi_1$ with probability $\sigma_2'^e := \norm{ T^e(\Delta')\varphi_1 }^2$. Thus the probability of the occurrence of an electron  after the position measurement  is
\begin{equation}\label{POE}
\sigma_2^e + \sigma_2'^e=1-2\langle \varphi_1,T^e(\Delta)T^e(\Delta') \varphi_1\rangle
\end{equation}
The  state $\phi_1$ also gives rise to the positron state $\varphi^p_2 =  \norm{  P^p T^\textsc{d}(\Delta)\varphi_1 }^{-1}  P^p T^\textsc{d}(\Delta)\varphi_1$ with probability $\sigma^p_2 =  \norm{ P^e T^\textsc{d}(\Delta)\varphi_1 }^2$. Moreover,  there is the positron state  $\varphi'^p_2 =  || P^p T^\textsc{d}(\Delta')\varphi_1 ||^{-1} P^p T^\textsc{d}(\Delta')\varphi_1$ with probability $\sigma'^p_2 =  \norm{ P^p T^\textsc{d}(\Delta')\varphi_1 }^2$. One easily checks that 
\begin{equation}\label{POP}
\sigma^p_2+\sigma'^p_2=2\langle \varphi_1,T^e(\Delta)T^e(\Delta') \varphi_1\rangle
\end{equation}
\hspace*{6mm} 
The foregoing considerations show that 
the influence of the apparatus represented by $T^\textsc{d}(\Delta)$ does not create eigenstates of $T^\textsc{d}(\Delta)$, which are not observable  but a mixed state of electron and positron states given by the density operator 
\begin{equation}\label{SO}
 W := \sum_{P,E}PEP_{\varphi_1} EP\quad \text{ with } P\in\{P^e,P^p\}, E\in\{T^\textsc{d}(\Delta), T^\textsc{d}(\Delta')\} 
\end{equation}
Here  $P_{\varphi_1}$ denotes  the projection operator on the subspace $\C\varphi_1$. Clearly $W\ge 0$ and, using the formula $QP_\phi Q=\norm{Q\phi}^2\norm{\phi}^{-2}P_{Q\phi}$ for any projection $Q$, one easily verifies $\operatorname{tr}(W)=1$. 
\\
\hspace*{6mm} 
 Let us review the formula (\ref{SO}). The eigenspaces $T^\textsc{D}(\Delta)\mathcal{H}^\textsc{D}$ and $T^\textsc{D}(\Delta')\mathcal{H}^\textsc{D}$ of the observable $T^\textsc{D}(\Delta)$ of the Dirac system are infinite dimensional, whence the eigenvalues  $0,1$ are degenerate. Without a detailed knowledge of the used measuring equipment one assumes a  {\it non-selective measurement}, which affects the state $\varphi_1$ in a minimal way. In this case one obtains for the state after the measurement the density operator $\tilde{W}:=\sum_E EP_{\varphi_1} E$ for $E\in\{T^\textsc{d}(\Delta), T^\textsc{d}(\Delta')\}$. However,  one does not have the information that the state after the position measurement belongs to one of the eigenspaces, as the states of the eigenspaces are not observable, but that the state is a mixture of an electron and an positron state. Hence the actual density operator follows $\sum_P P\tilde{W} P$ for  $P\in\{P^e,P^p\}$ as claimed.
 \\
\hspace*{6mm} 
Regarding the positron production the result is  that a non-selective measurement of the position of the electron in the state $\varphi_1$ in a spacetime region $\Delta$ creates a positron with probability
$$2\langle T^e(\Delta)\varphi_1,T^e(\Delta') \varphi_1\rangle$$
see (\ref{POP}). If the outcome "positron" is recorded then one infers  
from (\ref{SO})   that       the state of the positron     is represented by        the density  operator 
$ \frac{1}{\operatorname{Tr}(\tilde{W}P^p)}\, P^p\tilde{W}P^p$.

 \section{Recent related developments and outlook}}
We mention some related recent developments, although they concern
questions which are logically distinct from the 
causal-localization problem just described.  Halvorson and Clifton
\cite{HC01} analyzed no-go consequences of imposing, among other
assumptions, {\em commutativity requirements} on localization effects
associated with spacelike separated regions as prescribed in the Araki-Haag-Kastler approach to local quantum physics \cite{H96},\,\cite{A09}.  The relation between
particle localization, commutativity and the locality principle has
recently been reconsidered by one of the authors  \cite{M26a},\,\cite{M26b}. In the second paper in particular a causal localization has been constructed even for quantum fields and it is proved that commutativity can be restored when dealing with causally separated laboratories and local (conditional) localizations.
However in the first paper, following a type of analysis initiated by   Busch and collaborators \cite{Buschloc1,Buschloc2},   it is argued that, for causal localizations extended to the whole spacetime, the commutativity requirment is not necessary even if possible. 
A different further recent approach is due to Lechner and de Oliveira
\cite{LO26}, who construct causal quantum-mechanical localization
observables in lattices of {\em real} projections, in connection with {\em modular
localization}.  These works are related to the general problem of
relativistic localization, but they address different aspects of it sometime called microcausality.
\\
\hspace*{6mm}
The route followed here is instead the one initiated by the
causality condition and developed through achronal localization: the
construction of localization observables from conserved currents and
the analysis of their causal properties for the Dirac system and for
the positive-energy electron sector.
\\
\hspace*{6mm}
   In particular, \cite{M26b} is completely developed in a QFT scenario in order to make contact with the Araki-Haag-Kastler formulation, and also for systems with many particles, contrary to what we do in this work, where localization means localization of a single fermion. In that work, localization of a single boson is recovered when referring to single-particle states. This is the perspective from which we should consider the present work: localization of a fermion is expected to be the restriction to the one-particle space of a more general localization for the Dirac quantum field, or perhaps for the normally ordered current constructed with this quantum field operator. An investigation at the level of quantum field theory for the Dirac field, similar to that of \cite{M26b}, will be carried out elsewhere.

\section*{Acknowledgments}  C.D.R, and V.M acknowledge that this work was written within the activities of INdAM-GNFM.

\section*{Declaration statements}
{\bf Conflict of Interest}: The authors declare that they have no conflict of interest.\\
{\bf Ethical Statement}: This work does not involve human participants, animals, or sensitive data, and therefore no ethical approval was required.\\
{\bf Informed Consent}: Not applicable.\\
{\bf Data Availability}: No datasets were generated or analyzed in this study. All relevant information is contained within the article.\\
{\bf Funding}: This research received no external funding.

\appendix

 \section{Properties of RCL}

Recall that ${\cal E}({\cal H})$ denotes the space of the effects on a Hilbert space $\cal{H}$, i.e.,  the  bounded operators $A$ on  $\cal{H}$  with $0\le A \le I$.
We will see that an RCL  employs   the natural  {\bf monotone $\sigma$-complete effect algebra} (MSCEA) structures of $\mathcal{C}$ and  
${\cal E}({\cal H})$. 
\\
\hspace*{6mm}
An MSCEA  \cite{DP00} is an {\em effect algebra} $({\cal E}, \oplus, \perp, {\bf 0}, {\bf 1})$  such that, referring to the  canonical order $a\leq b$ iff there is $c \in {\cal E}$  with  $c\perp a$ and $a\oplus c = b$, every increasing sequence $(a_n)$  admits a supremum $\vee_n a_n \in {\cal E}$. $\oplus$ is called {\em partial sum} as it applies only to orthogonal summands. If $(a_n)$ is a countable  orthogonal sequence in ${\cal E}$ then put  $\oplus_n a_n:=\vee_n (a_1 \oplus \cdots \oplus a_n)$.
\\
\hspace*{6mm}
It is easy to see that every $\sigma$-complete orthomodular lattice $({\cal L}, \leq, \perp,  \vee, \wedge, {\bf 0}, {\bf 1})$, and hence in particular the causal logic $\mathcal{C}$, is a MSCSA
  when defining the  partial sum  $M_1 \oplus M_2 := M_1 \vee M_2$ for elements $M_1,M_2 \in \mathcal{L}$ with $M_1\leq  M_2^\perp$.  Also 
  ${\cal E}({\cal H})$ 
  has  a natural structure of  MSCEA  if defining $A^\perp := I-A$  and the partial sum 
$A\oplus B  := A+B$ for $A \leq B^\perp$.  
\\
\hspace*{6mm}
Let $\mathcal{E}_1$ and ${\cal E}_2$ be two MSCAE. A map $\phi : {\cal E}_1 \to {\cal E}_2$ is a  homomorphism  of 
 MSCEA if $\phi$ preserves the unities, i.e., $\phi(0)=0$ and $\phi(1)=1$, and the partial sums, i.e.,  $\phi(a\oplus b)=\phi(a)\oplus\phi(b)$, and satisfies 
$\phi\left( \vee_{n\in \mathbb{N} } a_n\right)=\vee_{n\in \mathbb{N}}  \phi \left( a_n \right)$  for  every increasing sequence $(a_n)$ in ${\cal E}$.

\begin{The}\label{HMSCAE}  Let $F: \mathcal{E}_1 \to {\cal E}_2$ a map, which preserves the unities and satisfies 
$$ F(\oplus_n a_n) = \oplus_n F(a_n)$$
for every countable orthogonal sequence $(a_n)$ in ${\cal E}_1$. Then $F$ is a homomorphism.
\end{The}

{\it Proof.} Let $a,b\in {\cal E}$ with $a\perp b$ and put
$
 c :=(a\oplus b)^\perp.
$
Then both
$
 (a,b,c,0,0,\ldots)
$
and
$
( a\oplus b,c,0,0,\ldots)
$
are orthogonal sequences with orthogonal sum $a\oplus b \oplus c \oplus 0\oplus \dots=(a\oplus b) \oplus c =1$. As $F$ preserves the unities by the hypothesis 
$
 1=(F(a)\oplus F(b))\oplus F(c) =F(a\oplus b)\oplus F(c)
$, whence $F(c)^\perp=F(a)\oplus F(b)$ and  $F(c)^\perp=F(a\oplus b)$. Therefore
$
 F(a\oplus b)=F(a)\oplus F(b).
$
Thus $F$ preserves the partial sum.
\\
\hspace*{6mm}
It remains to establish  preservation of suprema of increasing sequences. To this end,  let
$
 a_n\uparrow a
$. Put  $d_1:=a_1$ let $d_n$ satisfy $a_{n-1}\oplus d_n=a_n$ for $n\ge 2$.
Then obviously  $(d_n)$ is an orthogonal sequence and
$
 a_n=\oplus_{k=1}^{n}d_k
$
for every $n$. Moreover,
$$
 \oplus_{n}d_n=a
$$
because the supremum of the finite partial sums is precisely $\vee_n a_n=a$.
Therefore
$$F(\vee_na_n)=F(a)=F(\oplus_n d_n)=\oplus_n F(d_n)=\vee_n\oplus_{1}^nF(d_k) = \vee_n F(\oplus_{1}^nd_k) =\vee_nF(a_n)$$
\qed

As the range of an RCL is contained in ${\cal E}({\cal H})$,  one considers an  RCL  to be a map between  MSCAE. Actually, by the following corollary, {\em an} RCL {\em is an} MSCAE {\em homomorphism}.
In this context one may wonder whether there are representations of $F: {\cal C} \to {\cal L}({\cal H})$, i.e. homomorphisms,  with respect to the orthomodular lattice structure, evaluated in the lattice ${\cal L}({\cal H})$ of orthoprojectors of a separable Hilbert space ${\cal H}$. The answer is negative unless $F$ is trivial according to (127) Lemma in \cite{C17}.

 \begin{Cor}\label{PRCL} An RCL $F$ satisfies $F(\R^4)=I$,  $F(M^\perp)=I-F(M)$,  $F(\bigvee_n M_n)=\sum_nF(M_n)\le I$ for every sequence $(M_n)$ of mutually orthogonal sets, $F$ is finitely orthoadditive, and  is  monotone. Moreover, if $(M_n)$  is increasing, respectively decreasing, then $F(\bigvee_nM_n)=\lim_nF(M_n)$, respectively $F(\bigwedge_nM_n)=\lim_nF(M_n)$. 
\end{Cor}

{\it Proof.} Let $M\in\mathcal{C}$.  $F$ applies to the sequence $(M_n):= (M,M^\perp,\emptyset,\dots)$ yielding  $I=F(M)+F(M^\perp) +F(\emptyset)+\dots =F(M)+F(M^\perp)$, whence $F(M^\perp)=I-F(M)$. For $M=\emptyset$ it follows $F(\R^4)=I$. As to the next claim put $M:=\bigvee_n M_n$. Then $F$ applies to $(M^\perp, M_1,M_2,\dots)$ yielding $I=F(M^\perp)+\sum_nF(M_n)$, whence $F(M)=\sum_nF(M_n)\le I$. Given finitely many mutually orthogonal $M_1,M_2,\dots M_n$, $F$ applies to the sequence $(M_1,\dots,M_n,\emptyset,\dots)$, whence finite orthoadditivity of $F$ follows. 
 In order to show monotony, let $M\subset L$. By orthomodularity  $L= M\vee (M^\perp\wedge L)$, whence $F(L)=F(M)+F(M^\perp\wedge L)\ge F(M)$.
\\
\hspace*{6mm} 
Finally, in view of  (\ref{HMSCAE}), let $(M_n)$ be a sequence of mutually  orthogonal sets. Then obviously $F(\oplus_nM_n)=F(\bigvee_nM_n)=\sum_nF(M_n)$. Now one evaluates $\oplus_nF(M_n) = \vee_nA_n$ for  $A_n:= \sum_1^nF(M_k)$. As $(A_n)$ is a bounded  sequence (by $I$) of symmetric operators such that $(\langle x,A_nx\rangle)$ is increasing for every $x\in\mathcal{H}$, by 
\cite[Theorem 4.28]{W76} there is a bounded operator $A$ with $A_nx\to Ax$ for every $x\in\mathcal{H}$. {Clearly $0\le A_n\le A\le I$ for every $n$ and $B\ge A$ for every bounded operator $B$ with $B\ge A_n$ for all $n$. Hence $A=\sum_nF(M_n)$ strongly and $A=\vee_n A_n$, whence  $\oplus_nF(M_n) =\sum_nF(M_n)$.
\\
\hspace*{6mm} 
Thus the assumptions on $F$ in (\ref{HMSCAE}) are satisfied. So for every increasing sequence $(M_n)$ one has  $F(\bigvee_nM_n)=\bigvee_nF(M_n)$. The latter equals $\lim_nF(M_n)$ arguing for $(F(M_n))$ as above using \cite[Theorem 4.28]{W76}. If $(M_n)$ is decreasing then the claim holds by duality.\qed

 \section{Result on unions of graphs of $L$-Lipschitz functions.}

 Obviously a set $\Delta\subset \R^4$ is the graph of an $L$-Lipschitz function $\tau$ defined on $\varpi(\Delta)$  if and only if $|x_0-y_0|\le L|x-y|$ for $\mathfrak{x},\mathfrak{y} \in \Delta$. 

\begin{Lem}\label{LOAS} Let $L\ge 0$. Let $\Delta\subset \R^4$ be the union of 
sets $\Gamma$, which are the graphs of $L$-Lipschitz functions.
 Suppose that for every  $\mathfrak{a},\mathfrak{b} \in \Delta$,  $\mathfrak{a}\ne \mathfrak{b}$ there are $0=\lambda_0<\lambda_1<\dots<\lambda_{n-1} <\lambda_n=1$ for some $n\in\N$ such that, for $i=0,\dots,n-1$, $\{\mathfrak{x}^i,\mathfrak{x}^{i+1}\}$ lies in some $\Gamma$, where
$x^i:= a+\lambda_i (b    -  a) $ for $i=0,\dots,n$.
\\
\hspace*{6mm}
 Then  $\Delta$ is the graph of an $L$-Lipschitz function. If $L=1$ this means that  $\Gamma$ and $\Delta$ are achronal.
 If the sets $\Gamma$ are spacelike, then $\Delta$ is spacelike. 
  \end{Lem}\\ 
  {\itshape Proof.} Let  $\mathfrak{a},\mathfrak{b} \in \Delta$,  $\mathfrak{a}\ne \mathfrak{b}$  and $\lambda_i$ as above. Then $|a_0-b_0|\le \sum_{i=0}^{n-1} |x^i_0-x^{i+1}_0 |\le  \sum_{i=0}^{n-1}  L |x^i -x^{i+1}| = L |a-b|$ since the points $x^i$ lie on the straight line joining $a$ and $b$.
  \\
\hspace*{6mm}
 If the sets $\Gamma$ are spacelike, then the inequality holds for $L=1$, i.e.,  $|a_0-b_0|\le |a-b|$. Since $\mathfrak{a}\ne \mathfrak{b}$, this implies $a\ne b$ and hence $x^i\ne x^{i+1}$. Now, due to the  spacelikeness of $\Gamma$, the inequality holds even for $<$ in place of $\le$. \qed

See the applications  (\ref{LCLL})\,-\,(\ref{INFLF}),  (\ref{GPGL}), (\ref{GMMAS}) of (\ref{LOAS}). The result (\ref{LCLL})   concerns locally Lipschitz functions used in (\ref{RPLFPOL}), and  (\ref{INFLF}) on the infimum of two Lipschitz functions is used in (\ref{MTPVAL}).

\begin{Lem}\label{LCLL} Let $L>0$ and let $D\subset \R^3$ be convex. Let $\tau:D\to\R$ be locally Lipschitz with  local Lipschitz constants $L'<L$.
Let $K\subset D$ be compact. Then $\tau|_K$ is $L'$-Lipschitz for some constant $L'<L$.
Moreover $\tau$ is $L$-Lipschitz.
\end{Lem}
\\
 {\itshape Proof.} As known the convex hull of $K$ is compact. It is contained in $D$. Hence  it is no restriction to assume that $K$ is convex.  By assumption for every $x\in K$ there is an open set $U_x$ containing $x$ such that $\tau|_{D\cap U_x}$ is $L_x$-Lipschitz for some $0\le L_x<L$. As $K$ is compact there are finitely many $U_l:=U_{x_l}$ covering $K$. Put $L':=\max \{L_{x_l}\}<L$. 
\\
\hspace*{6mm} 
For $a,b\in K$, $a\ne b$  the segment $[a,b]:=\{a+\lambda (b-a): \lambda\in[0,1]\}$ is contained in $ K$.  
For  $x\in[a,b]$ there is $l$ with $x\in U_l$. Let $B_x\subset U_l$ be an open ball with center $x$. Since   $[a,b]$  is compact there are finitely many $B_k:=B_{x_k}$ covering $[a,b]$. 
Note that $\tau|_{\overline{B}_k\cap[a,b]}$ is $L'$-Lipschitz.
Now (\ref{LOAS}) applies. Indeed, the boundary points of the balls at which  $[a,b]$  leaves a ball $B_k$ as $\lambda$ increases determine the values $\lambda_i$. Hence $\tau|_K$ is $L'$-Lipschitz.
\\
\hspace*{6mm} 
Finally let $a,b\in D$, $a\ne b$. As $[a,b]\subset D$ is compact, the previous result applies. Hence $|\tau(a)-\tau(b)|\le L'|a-b|$  for some  $L'<L$. Therefore $\tau$ is $L$-Lipschitz.
\qed

\begin{Lem}\label{PLLF}  Let $L\ge 0$.  Let $\Delta\subset \R^4$ be the union of two closed 
sets $\Gamma_1, \Gamma_2$, which are the graphs of $L$-Lipschitz functions. Suppose that  $\varpi(\Delta)$ is convex. Then $\Delta$ is the graph of an $L$-Lipschitz function. If $\Gamma_1$, $\Gamma_2$ are spacelike then so is $\Delta$.
\end{Lem}
\\
 {\itshape Proof.}
 Let  $\mathfrak{a},\mathfrak{b} \in \Delta$,  $\mathfrak{a}\ne \mathfrak{b}$.  Without restriction $\mathfrak{a} \in \Gamma_1$.
 Put $x^\lambda:= a+\lambda(b-a)$ for $\lambda\in [0,1]$.  Then  $x^\lambda \in\varpi(\Delta)$ and $x^0\in  \varpi(\Gamma_1)$.
 Let $\lambda':=\sup\{\lambda\in [0,1]: x^\lambda\,\in \varpi(\Gamma_1)\}$.
\\
\hspace*{6mm} 
(a) Show $\big(\tau(x^{\lambda'}),x^{\lambda'}\big)\in \Gamma_1$. Indeed, there are $\lambda_n\uparrow \lambda'$ such that $x^{\lambda_n} \in\varpi(\Gamma_1)$ and $x^{\lambda_n}\to\ x^{\lambda'}$. Then $ \Gamma_1 \ni \big(\tau(x^{\lambda_n}),x^{\lambda_n}\big) \to
\big(\tau(x^{\lambda'}),x^{\lambda'}\big)\in \Gamma_1$ as $\tau$ is continuous and $\Gamma_1$ closed.
\\
\hspace*{6mm} 
(b) If $\lambda'=1$ then $\mathfrak{b} \in \Gamma_1$, whence $|a_0 -b_0|=|\tau_1(a)-\tau_1(b)|\le L|a-b|$  for $\operatorname{graph}(\tau_1)=\Gamma_1$. If $\lambda'<1$, let $\lambda \in ]\lambda',1]$. Then 
$x^\lambda \not\in \varpi(\Gamma_1)$, whence  $x^\lambda \in \varpi(\Gamma_2)$. Therefore, arguing as in (a), it follows $\big(\tau(x^{\lambda'}),x^{\lambda'}\big)\in \Gamma_2$.
Now (\ref{LOAS}) applies.\qed

\begin{Cor}\label{INFLF} Let $L\ge 0$ and let $D\subset \R^3$ be closed convex. Let $\tau_i$ be $L$-Lipschitz on $D$ and put $\Lambda_i=\operatorname{graph}(\tau_i)$,  $i=1,2$. Then  $\tau:=\inf\{\tau_1,\tau_2\}$ is $L$-Lipschitz, and if $\Lambda_i$, $i=1,2$ are spacelike, then so is $\Lambda:=\operatorname{graph}(\tau)$. The same holds true for  $\tau:=\sup\{\tau_1,\tau_2\}$. 
\end{Cor}
\\
 {\itshape Proof.}  $\Gamma_i:=\{\mathfrak{x}\in \Lambda: \tau(x)=\tau_i(x)\}$, $i=1,2$ are closed  with $\Lambda=\Gamma_1\cup \Gamma_2$, $\varpi(\Lambda)=D$ convex. 
$\Gamma_i$ is the graph of  a $L$-Lipschitz function, respectively  $\Gamma_i$ is spacelike. Hence (\ref{PLLF}) applies.\qed

\section{Pyramids with special polyhedronal basis}

Let $B:=\{x\in R^3: |x|\le 1\}$ be the unit ball and $\partial B$ the unit sphere in $\R^3$. Let $\gamma > 0$ and $v\in\partial B$. Then $H_v:=\{\mathfrak{x}\in\R^4: x_0=\gamma xv\}$
is the hyperplane in $\R^4$ containing the origin which touches  $\{\gamma\}\times B$ at $(\gamma,v)$.  The closed half-space $H'_v :=\{\mathfrak{x}\in\R^4: x_0\ge \gamma xv\}$ contains  $\{\gamma\}\times B$.

\begin{Lem}\label{CSP} Let $F\subset \partial B$ be finite. Obviously $\bigcap_{v\in F}H'_v$ is a closed convex set containing $\{\gamma\}\times B$ and satisfying $\bigcap_{v\in F}H'_v \cap (\{\gamma\}\times \R^3)  = \{\gamma\}\times Q$ for $Q:=\bigcap_{v\in F}\{x\in\R^3:xv\le 1\}$.
Then 
\begin{itemize}
\item[\emph{(a)}] $Q$ is a closed convex polyhedron if there are  linearly independent $v_1,v_2,v_3$ in $\R^3$ such that 
$\pm v_1,\pm v_2,\pm v_3\in F$. $B$ touches every face $Q\cap \{x\in\R^3:xv=1 \}$
of $Q$ in $v\in F$.
\item[\emph{(b)}] Let $q>1$. Then there is a finite $F\subset \partial B$ such that $Q\subset qB$.

\end{itemize}
\end{Lem}

{\itshape Proof.} (a) Let 
$x\in Q\setminus \{0\}$.
Then $1\ge\pm xv_i$, whence $|xv_i|\le 1$,  $i=1,2,3$. Let $(e_1,e_2,e_3)$ be the standard basis of $\R^3$ and $v_j=\sum _{i=1}^3A_{ij} e_i$.  Note $\sum_{i=1}^3|\hat{x}e_i|^2=1$ for   $\hat{x}:=\frac{1}{|x|}x$.
Then by linear analysis there is $j\in\{1,2,3\}$ such that $(\hat{x} v_j)^2\ge\alpha^2/3$ with $\alpha$ the smallest eigenvalue of $|A|$. This implies $|x|\le \sqrt{3}/\alpha$. Hence $Q$ is bounded. The result follows.
\\
\hspace*{6mm}
(b) For $v\in \partial B$ let $U_v:=\{y\in \partial B:yv>\frac{1}{q}\}$. Since $v\in U_v$, $(U_v)_v$ is an open covering of $\partial B$. Hence there is a finite  $F\subset \partial B$ such that $\partial B=\bigcup_{v\in F}U_v$.  Let  $F$ also contain  $\pm v_1,\pm v_2,\pm v_3$ with  $v_1,v_2,v_3$  linearly independent. Recall (a). Let $Q$ be determined by $F$.
\\
\hspace*{6mm}
Now let $x\in Q\setminus\{0\}$. Then $xv\le 1$ for $v\in F$. For  $\hat{x}=\frac{1}{|x|}x$  there is $v\in F$ such that $\hat{x}\in U_v$. Hence $\hat{x}v>\frac{1}{q}$. It follows $|x|< q$.\qed

\begin{Pro}\label{GPGL} Let $\gamma>0$. Let $P\subset \R^4$ be the infinite pyramid with vertex $0$ and generated by the future directed straight half-lines joining the vertex with the points of  $\{\gamma\}\times Q$, where $Q$ is a polyhedron defined in \emph{(\ref{CSP})}. Then $\partial P$ is the graph of a $\gamma$-Lipschitz function on $\R^3$. In particular, if $\gamma\le 1$, then $\partial P$ is maximal achronal.
\end{Pro}\\
 {\itshape Proof.} Put   $Q_t:= \frac{t}{\gamma}Q$, $t>0$.  Then  $\partial P\cap (\{t\}\times \R^3) =\{t\}\times \partial Q_t$ and $Q_t\supset \frac{t}{\gamma} B$. 
So one infers that $\varpi$ maps $\partial P$ bijectively onto $\R^3$.   
   \\
\hspace*{6mm} 
 The hyperplanes  $H_v:=\{\mathfrak{x}\in\R^4: x_0=\gamma xv\}$, $v\in F$
are     the tangent spaces of $P$ off the edges.     
    \\
\hspace*{6mm} 
 Now note that $H_v$ is the graph of the $\gamma$-Lipschitz function $x\mapsto \gamma xv$ and that $\partial P$ is the union of the sets  $H_v\cap \partial P$. Every edge of $\partial P$ is mapped by $\varpi$ on a plane $E_w:=\{x\in\R^3:xw=0\}$, where $w$ equals $v-v'$ for some $v,v'\in F$, $v\ne v'$.
     \\
\hspace*{6mm} 
So (\ref{LOAS}) applies. Indeed, let  $\mathfrak{a},\mathfrak{b} \in \partial P$,  $\mathfrak{a}\ne \mathfrak{b}$ and let $[a,b]$ be the segment $\{a+\lambda (b-a): 0\le\lambda\le 1\}$. Then the  points at which $[a,b]$ crosses a plane $E_w$ determine the values $\lambda_i$.\qed 

\begin{Def}\label{GGIP}      More generally we consider  the spacetime translated pyramids $\mathfrak{a}+P$. Their characteristics are the position of the \textbf{vertex},  the  \textbf{steepness}  $\gamma>0$ and the \textbf{bound} $q>1$ of the pyramid.
\end{Def}

\begin{Lem} \label{GPCLPM} Let $P$  be the pyramid  \emph{(\ref{GPGL})} with vertex $0$, steepness $\gamma$ and bound $q$. Then     
 \\
\hspace*{6mm}  
\begin{equation}\label{GCIP}
\{\mathfrak{x}: x_0\ge \gamma|x|\}\subset\;    P\;\subset \{\mathfrak{x}: x_0\ge \frac{\gamma}{q}|x|\}
\end{equation}
\end{Lem}
\\
{\itshape Proof.}  Let $\mathfrak{x}=(\gamma,x) \in C:=\{\mathfrak{y}: y_0\ge \gamma|y|\}$. Then $|x|\le 1$, whence $\mathfrak{x}\in P$. Therefore $C\subset P$.
Now let  $C:=\{\mathfrak{y}: y_0\ge \frac{\gamma}{q}|y|\}$. Obviously $0\in C$.
Let $\mathfrak{x}\in \{\gamma\}\times Q$. Then $|x|\le q$. Hence $\mathfrak{x}\in C$. So the straight half-line joining the vertex $0$ with  $\mathfrak{x}$ is contained in $C$. Therefore  $P\subset C$. \qed

\begin{Lem}\label{GCIMAS}
Let $0\le L<L'$. Let $\Lambda$ be the graph of an $L$-Lipschitz function  $\tau$ on $\R^3$ and $\Lambda^-:=\{\mathfrak{x}:x_0\le \tau(x)\}$.  Let $\mathfrak{y}\in\Lambda$ and $\alpha\ge 0$.   Then the set $\Lambda^- \cap C$ with
$$C:= \{\mathfrak{x}: x_0-y_0 +\alpha\ge L' |x-y|\}$$ 
contains $\mathfrak{y}$ and is compact. Moreover
\begin{equation}\label{GIELC}
\Big\{x\in\R^3: |x-y|\le \frac{\alpha}{L'+L} \Big\}\subset \;\varpi(\Lambda\cap C)\;\subset \Big\{x\in\R^3: |x-y|\le \frac{\alpha}{L'-L} \Big\}
\end{equation}
\end{Lem}\\
{\itshape Proof.}    Obviously, $\mathfrak{y}\in C$ and
$\Lambda^- \cap C$ is closed. Let $\mathfrak{x} \in \Lambda^- \cap C$. Then $x_0\le\tau(x) \le L|x|+\tau(0)$ and $x_0\ge L'|x-y| - \alpha+y_0 \ge L'|x|-L'|y| - \alpha+y_0$. Hence 
$$  L'|x|-L'|y| - \alpha+y_0\le x_0 \le  L|x|+\tau(0)$$
 whence $|x|\le (2L'|y|+\alpha)/(L'-L)$.  Therefore $ \Lambda^- \cap C$ is bounded.
     \\
\hspace*{6mm} 
Let $x\in \varpi(\Lambda\cap C)$. Put  $x_0:=\tau(x)$ and $\mathfrak{z}:=\mathfrak{x} - \mathfrak{y}$. Then $-L|z|\le z_0\le L|z|$ and $z_0\ge L'|z|-\alpha$ hold, whence  $ L'|z|-\alpha\le L|z|$, i.e.,  $|z|\le  \frac{\alpha}{L'-L}$. Now let $x\in\R^3$.  Put  $x_0:=\tau(x)$ and $\mathfrak{z}:=\mathfrak{x} - \mathfrak{y}$, and assume $|z|\le  \frac{\alpha}{L'+L}$. Then $ L'|z|-\alpha\le -L|z| \le z_0$. Hence (\ref{GIELC}).\qed

\begin{Cor}\label{GCCIMAS} As to $\Lambda$ see $\emph{(\ref{GCIMAS})}$. Let $P$ be  the pyramid with vertex $(y_0-\alpha,y)$,  bound $q$, steepness $\gamma> qL$.  Then 
\begin{equation}
\Big\{x\in\R^3: |x-y|\le \frac{\alpha}{\gamma+L} \Big\}\subset \;\varpi(\Lambda\cap P)\;\subset \Big\{x\in\R^3: |x-y|\le \frac{\alpha}{  \gamma/q -L} \Big\}
\end{equation}
\end{Cor}\\
{\itshape Proof.}   Apply (\ref{GIELC}) to the cone contained in $P$ and that containing $P$ from (\ref{GCIP}).\qed

\begin{Pro}\label{GMMAS} Let $L\ge 0$. Let $\Lambda$ be the graph of an $L$-Lipschitz function  $\tau$ on $\R^3$ and $\Lambda^-:=\{\mathfrak{x}:x_0\le \tau(x)\}$.  Let $\mathfrak{y}\in\Lambda$ and $\alpha\ge 0$.  Let $P$ be  the pyramid with vertex $(y_0-\alpha,y)$, bound $q$, steepness $\gamma> qL$. 
     \\
\hspace*{6mm}
Then the sets  $\Lambda^- \cap P$ , $\Lambda_P:=\Lambda\cap P$ contain $\mathfrak{y}$ and are compact,
 $\varSigma_P:= \partial P\cap \Lambda^-$ contains  $(y_0-\alpha,y)$ and is compact, and
 \begin{equation}\label{GSML}
\Lambda':=(\Lambda\setminus \Lambda_P)\cup \varSigma_P
\end{equation}
is the graph of a $\gamma$-Lipschitz function on $\R^3$. If $\gamma\le 1$ then $\Lambda'$
 is maximal achronal.
\end{Pro}\\
{\itshape Proof.} 
 By  (\ref{GCIMAS}) for $L'=\gamma/q$ and by (\ref{GCIP}), $\Lambda^- \cap P$ is compact. $\Lambda_P$, $\varSigma_P$ are closed subsets of  $ \Lambda^- \cap P$.
    \\
\hspace*{6mm} 
As to (\ref{GSML}), one checks $\varpi(\varSigma_P)=\varpi(\Lambda_P)$. Indeed, let $\mathfrak{x}\in \varSigma_P$.  Since   $(x_0,x)\in \partial P$ and $x_0\le \tau(x)$, it follows $(\tau(x),x)\in P$. Conversely, let  $\mathfrak{x}\in \Lambda_P$. Then $x_0=\tau(x)$ and  $(\tau(x),x)\in P$, whence there is $\xi\le \tau(x)$ such that $(\xi,x)\in \partial P$.
\\
\hspace*{6mm}
It follows $\varpi(\Lambda')=\R^3$.  It remains to show that $\Lambda'$ is achronal \cite[(1)\,(f)]{C24}.
\\
\hspace*{6mm}
$\tau$ is $\gamma$-Lipschitz as $L<\gamma$.   Hence $\Lambda$  and, by (\ref{GPGL}), also $\partial P$ is the graph of an $\gamma$-Lipschitz function.  So (\ref{LOAS})  applies to the union $\Lambda'=\overline{\Lambda\setminus \Lambda_P} \cup \varSigma_P$.\qed

\section{When the localization operators are projection operators\,?}

Let $0\le L\le 1$.
A Borel set $\Delta\subset \R^4$ is called a (maximal) $ L$-achronal set if it is the graph of a $ L$-Lipschitz function $\tau$ (on $\R^3$). Obviously the class of these Borel sets 
 is contained in $\mathcal{B}^{ach}$  and it  is Poincar\'e invariant. It contains all flat spacelike Borel sets. Note also that every $L$-achronal set is subset of a maximal $L$-achronal set.
This fact follows easily from the case $L=1$ referring to achronal sets  \cite[(1)(c)]{C24}. Let $\Lambda$ be maximal acronal.
 Keep in mind that $\varpi|_\Lambda$ is an homeomorphism from $\Lambda$ on $\R^3$. Actually  $|x-y| \le |\mathfrak{x} - \mathfrak{y}| \le \sqrt{2}\, |x-y|$ for $ \mathfrak{x},\mathfrak{y} \in \Lambda$.
 
\begin{The}\label{PLFPOL}  To every   $0\le L < 1$ and  $ L$-achronal Borel set $\Delta$ a nonnegative operator  $T(\Delta)$ is assigned such that 
$T(\Delta)=0$ if $\mathcal{L}^3(\varpi(\Delta))=0$ and such  that $T$ is a positive operator valued normalized measure on 
every maximal $ L$-achronal Borel set. 
\\
\hspace*{6mm}
\emph{(a)} There is a unique extension $\tilde{T}$ of $T$ by nonnegative operators $\tilde{T}(\Delta)$ assigned to achronal Borel sets $\Delta$, whose bounded subsets  are  $L$-achronal with constants $ L<1$, such that  $\tilde{T}$  is $\sigma$-additive. One has $\tilde{T}(\Delta)=0$ if $\mathcal{L}^3(\varpi(\Delta))=0$.
\\
\hspace*{6mm}
\emph{(b)} Suppose that $T(\Delta)$ is a projection operator for every flat spacelike Borel set $\Delta$. Then $\tilde{T}$ is projection operator valued.
\end{The}

\begin{Rem}\label{RPLFPOL} In view of (\ref{PLFPOL})(a) consider a maximal achronal $\Lambda$ being locally $L$-achronal  with constants $L<1$. Then every Borel $\Delta\subset \Lambda$ satisfies the condition in (\ref{PLFPOL})(a), i.e., it is acronal and every bounded  subset is $L$-achronal for some $L<1$. Indeed, apply (\ref{LCLL}) to $D=\R^3$ for $L=1$.
\end{Rem}

The proof of  (\ref{PLFPOL}) uses the following decisive  technical result  (\ref{MTR}).

\begin{Pro}\label{MTR} Let $\Lambda$ be a maximal achronal set and let $T$ be a positive operator valued not necessarily normalized measure for the Borel sets of $\Lambda$. Assume $T(\Delta)=0$ if $\mathcal{L}^3(\varpi(\Delta))=0$.
Further, for $x\in\R^3$ and  $n\in\N$ let $C_{x,n} \subset \R^3 $ be  compact satisfying
\begin{itemize}
\item[\emph{(i)}]  $x\in C_{x,n} $
\item[\emph{(ii)}]  $0<\operatorname{diam}(C_{x,n}) \to 0$ as $n\to \infty$
\item[\emph{(iii)}]  $\operatorname{diam}(C_{x,n})^3\le C \,\mathcal{L}^3(C_{x,n})$ for all $x,n$ for some constant $C<\infty$
\end{itemize}
Then $T(\Delta)$ is a projection operator for every Borel  set $\Delta\subset \Lambda$,
if $T( \varpi^{-1}(C_{x,n})\cap\Lambda) $ is a projection operator for all $x,n$.

\end{Pro}
{\itshape Proof.} (a) For  open bounded  $E\subset \R^3$ let $\mathcal{C}_E:=\{C\subset E: C=C_{x,n} \text{ for some } x,n\}$. Then $\mathcal{L}^3(E)<\infty$, and  for every $x\in E$ and $\delta>0$ there is $C\in\mathcal{C}_E$ such that   $x\in C$  and $0<\operatorname{diam}(C)<\delta$  by (i), (ii). Hence due to (iii), the Vitali covering theorem by Lebesgue \cite{L1910} applies. Accordingly there are countably many mutually disjoint sets $C_j$ from $\mathcal{C}_E$ such that $\mathcal{L}^3(E\setminus \bigcup_jC_j)=0$.
 \\
\hspace*{6mm}
(b) Introduce the image measure  $S:=\varpi|_\Lambda(T)$  on $\R^3$ of $T$ on $\Lambda$, i.e.,  $S(E)= T(\varpi^{-1}(E)\cap \Lambda) $. Put $\mathcal{D}:=\{E\subset \R^3: E \text{ Borel, }  S(E)  \text{ projection operator} \}$. $\mathcal{D}$ is a Dynkin system. Indeed, $\R^3\in \mathcal{D}$. If $D, E \in\mathcal{D}$  with $D\subset E$ then $E\setminus D \in \mathcal{D}$ since 
 $S(E\setminus D)=  S(E) - S(D)$, where  $S(E), S(D)$ are projection operators with  $S(E) S(D) = S(D)$. And finally, if $E_j\in \mathcal{D}$, $j\in\N$ are mutually disjoint, then $\bigcup_jE_j\in\mathcal{D}$. Indeed, $S(\bigcup_jE_j)=\lim_{N\to \infty} S(\bigcup_{j=1}^N E_j)$ and  $ S(\bigcup_{j=1}^N E_j)=\sum_{j=1}^N S(E_j)$ is an projection operator, since   $S(E_j)$ are projection operators with $S(E_j)S(E_k)=0$ for $i\ne k$. Hence $  (S(\bigcup_{j=1}^N E_j))_N$  is an increasing sequence of projection operators, whence the claim by 
\cite[Theorem 4.32]{W76}.
 \\
\hspace*{6mm}
(c)  $\mathcal{O}:=\{E\subset \R^3: E \text{ open bounded}\}$ is $\cap$-stable. Therefore the Dynkin system $\mathcal{D}(\mathcal{O})$ generated by $\mathcal{O}$ is a $\sigma$-algebra. Hence  $\mathcal{D}(\mathcal{O})$ equals the set of Borel sets $\mathcal{B}(\R^3)$.
 \\
\hspace*{6mm}
(d) Now suppose that $S(C_{x,n})$ is a projection operator  for every $x, n$  and let $E\in \mathcal{O}$. Then by (a) and (b), $E=A\cup B$, where $A,B$ are disjoint Borel sets with $\mathcal{L}^3(A)=0$ and $S(B)$ a projection operator. Since $A=\varpi(\Delta)$ for $\Delta:=\varpi^{-1}(A)\cap\Lambda$, $T(\Delta)=0$ by the assumption on $T$. This means $S(A)=0$. Hence $S(E)=S(B)$ is a projection operator, whence $E\in\mathcal{D}$. So  $\mathcal{O}\subset \mathcal{D}$. It follows $\mathcal{D}(\mathcal{O})\subset \mathcal{D}$. By (c) this implies that $S$  is projection operator valued, whence the result.\qed

\begin{P8T}\label{P8T}  (a) Put $C_M:=\varpi^{-1}(\{x\in\R^3:|x|\le M\})$ for  $M\in\N$. Let $\Delta\in\mathcal{B}'$, which means that the bounded subsets of $\Delta\in\mathcal{B}^{ach}$ are $L$-achronal with constants $ L<1$. Define 
\begin{equation}
\tilde{T}(\Delta):=\lim_{M\to\infty}T(\Delta\cap C_M)\tag{*}
\end{equation}
 Note that the strong limit exists and defines a nonnegative operator $\le I$ by \cite[4.28(b)]{W76}. In case that $\Delta$ is  $L$-achronal, then evidently $\tilde{T}(\Delta)=T(\Delta)$, whence $\tilde{T}$ is an extension of $T$. Conversely, if  $\tilde{T}$ is a $\sigma$-additive  extension of $T$, then  $\tilde{T}$ satisfies (*) proving the uniqueness of a possible extension. Furthermore, since $T$ vanishes at null sets of $\mathcal{L}^3\circ \varpi$, so does $\tilde{T}$ due to (*).
 \\
\hspace*{6mm}
As to the $\sigma$-additivity of $\tilde{T}$ given by (*), let $\Delta_n\in\mathcal{B}'$, $n\in\N$, be mutually disjoint such that $\Delta\in\mathcal{B}'$ for $\Delta:=\bigcup_n\Delta_n$.
 \\
\hspace*{6mm}
Then $\sum_{n=1}^N \tilde{T}(\Delta_n)=\lim_{M\to\infty} \sum_{n=1}^N T(\Delta_n\cap C_M) = \lim_{M\to\infty}  T(\bigcup_{n=1}^N\Delta_n\cap C_M) = \tilde{T}(\bigcup_{n=1}^N\Delta_n)$, whence using \cite[4.28(b)]{W76} $\sum_{n=1}^\infty \tilde{T}(\Delta_n) =\lim_{N\to\infty} \tilde{T}(\bigcup_{n=1}^N\Delta_n)$. It remains to show that the foregoing limit equals $\tilde{T}(\Delta)$.
 \\
\hspace*{6mm}
For every state $\phi$ consider the   double sequence $(\langle \phi,T\big(\bigcup_{n=1}^N(\Delta_n) \cap C_M)\big) \phi\rangle)_{N,M}$. It is increasing. Hence the order of the limits $N\to \infty$ and $M\to \infty$ can be interchanged. One finds $\langle \phi, \tilde{T}(\Delta)\phi\rangle=
 \langle \phi,\lim_{N\to \infty} \tilde{T}(\bigcup_{n=1}^N\Delta_n)\phi\rangle$, whence the claim.
\\
\hspace*{6mm}
(b)  It suffices to show that $T$ is projection operator valued. Then the claim follows from  (*) by \cite[4.32]{W76}.
\\
\hspace*{6mm}
Let $\Lambda$ be maximal $L$-achronal with $\Delta\subset \Lambda$.  Let $\mathfrak{y}\in\Lambda$, $n\in\N$. For the definition of a pyramid see (\ref{GGIP}). Let $P_n$ be  the pyramid with vertex $(y_0-\frac{1}{n},y)$, bound $1< q< 1/L$ and steepness $1>\gamma> qL$.  Put $C_{y,n}:=\varpi(\Lambda \cap P_n)$. Clearly  $C_{y,n}$ is closed and $y\in C_{y,n}$. Obviously (\ref{MTR})(ii),(iii)  hold by (\ref{GCCIMAS}) with $C=\frac{6}{\pi}\big(\frac{1+L/\gamma}{1/q-L/\gamma}\big)^3$.
\\
\hspace*{6mm} 
According to  (\ref{MTR}) it suffices to show that $T(\Lambda_n) $    for $\Lambda_n:= \varpi^{-1}(C_{x,n})\cap\Lambda=\Lambda\cap P_n$ is a projection operator. By (\ref{GMMAS}), $\Lambda':=(\Lambda\setminus \Lambda_n) \cup \varSigma_n$ for $\varSigma_n:=\partial P_n \cap \Lambda^-$ is a maximal $\gamma$-achronal set. Evidently $\Lambda'\setminus \varSigma_n=\Lambda\setminus\Lambda_n$. So $I=T(\Lambda_n)+T(\Lambda\setminus\Lambda_n)$ and $I=T(\varSigma_n)+T(\Lambda'\setminus \varSigma_n)$ imply $T(\Lambda_n)=T(\varSigma_n)$. It remains to mention that $\partial P_n$ is the union of finitely many disjoint flat Borel sets.\qed

\end{P8T}

\begin{P7T}\label{P7T}
Every achronal set  is contained in a maximal achronal set $A$. So one is going to show that $T$ is projection operator valued on $A=\operatorname{graph}(\tau)$ with $\tau$ a $1$-Lipschitz function on $\R^3$. Again (\ref{MTR}) is decisive.
\\
\hspace*{6mm}
Let $\mathfrak{y} \in A$, $n \in \N$, $0<\gamma <1$. Put $\Lambda_n:= A_n\cup L_n$, where 
$A_n:=\{\mathfrak{x} \in A:\tau(x)\le y_0+\frac{1}{n}\}$ and $L_n:=\{\mathfrak{x}\in \R^4: x_0=y_0+\frac{1}{n}, \tau(x)>y_0+\frac{1}{n}\}$ are disjoint.
\\
\hspace*{6mm}
(a) $\Lambda_n$ is maximal achronal. Indeed, $\Lambda_n$ is the graph of the infimum of $\tau$ and the constant $y_0+\frac{1}{n}$, both $1$-Lipschitz functions on $\R^3$. Hence the claim holds by (\ref{INFLF}). So $T(\Lambda_n)=I$.
\\
\hspace*{6mm}
(b) The spacelike cone $\Gamma_n:=\{\mathfrak{x}:  x_0-y_0+\frac{1}{n} =\gamma|x-y|\}$ with vertex $(y_0-\frac{1}{n},y)$ and steepness $\gamma$ is maximal $\gamma$-achronal  
as it is the graph of the $\gamma$-Lipschitz function $x\mapsto \gamma|x-y|+y_0+\frac{1}{n}$ on $\R^3$. Hence, by (\ref{PLFPOL}), $T$ is projection operator valued on $\Gamma_n$.
\\
\hspace*{6mm} 
(c)  Let $\Lambda_n^-$ denote the closed lower half-space with boundary $\Lambda_n$ and $\Gamma_n^+$ the closed upper half-space with boundary $\Gamma_n$. The claim is that
\begin{equation}
 \Lambda'_n:=(\Lambda_n\setminus \Gamma_n^+) \cup (\Gamma_n\cap\Lambda_n^-)  \tag{*}
 \end{equation}
is maximal acronal.   First check $\varpi(\Lambda_n\cap \Gamma_n^+) =\varpi(\Gamma_n\cap\Lambda_n^-)$. Let $\mathfrak{x}\in \Lambda_n \cap \Gamma_n^+$. Then 
 there is $x_0'\le x_0$ such that $(x_0',x)\in\Gamma_n$. Hence $(x_0',x) \in \Gamma_n\cap\Lambda_n^-$. Conversely let  $\mathfrak{x}\in \Gamma_n\cap \Lambda_n^-$. Then there is $x_0'\ge x_0$ such that $(x_0',x)\in \Lambda_n \cap \Gamma_n^+$.
\\
\hspace*{6mm}
It follows $\varpi(\Lambda_n')=\R^3$.  So it remains to show that $\Lambda_n'$ is achronal \cite[(1)\,(f)]{C24}. $\Lambda_n$ and $\Gamma_n$ are $1$-Lipschitz.
So (\ref{PLLF})  applies to the union $\Lambda_n'=\overline{\Lambda_n\setminus \Gamma_n^+} \,\cup (\Gamma_n\cap\Lambda_n^-)$.
\\
\hspace*{6mm}
(d) Recall that $\Lambda_n'$ is maximal achronal. As the sets of the union  (*) are disjoint one has $I = T(\Lambda_n')= T(\Lambda_n\setminus \Gamma_n^+) + T(\Gamma_n\cap\Lambda_n^-)$. Since $I = T(\Lambda_n) = T(\Lambda_n\setminus \Gamma_n^+) +  T(\Lambda_n \cap \Gamma_n^+)$ it follows $T(\Gamma_n\cap\Lambda_n^-) = 
 T(\Lambda_n \cap \Gamma_n^+)$. Therefore $ T(\Lambda_n \cap \Gamma_n^+)$ is a projection operator by (b). Further  $ T(\Lambda_n \cap \Gamma_n^+)= T(A_n \cap \Gamma_n^+) + T(L_n \cap \Gamma_n^+)$ by (a). $L_n$ being flat, $T(L_n \cap \Gamma_n^+)$ a projection operator. One concludes that $T(A_n \cap \Gamma_n^+) $ is a projection operator.
\\
\hspace*{6mm}
(e)  $A_n \cap \Gamma_n^+$ is bounded and hence compact as closed. Indeed, put $\mathfrak{z}:=  \mathfrak{x}- \mathfrak{y}$.  Then  $\mathfrak{x}\in A_n \cap \Gamma_n^+$  means $ \mathfrak{x}=(\tau(z+y),z+y)$  and  $z_0\le \frac{1}{n}$, $z_0+\frac{1}{n} \ge \gamma|z|$. It follows $|z|\le\frac{2}{\gamma \,n}$ and $|z_0|\le \frac{1}{n}$.
\\
\hspace*{6mm}
(f) Therefore $C_{y,n}:=\varpi(A_n \cap \Gamma_n^+)$ is compact and according to (\ref{MTR}) it remains to show that  the sets  $C_{y,n}$ satisfy (\ref{MTR}) (i),(ii),(iii). Clearly $y\in C_{y,n}$ and  $\operatorname{diam}(C_{y,n})\le  \frac{4}{\gamma \,n}$. 
\\
\hspace*{6mm}
Turn to  (\ref{MTR})\,(iii). Let $z\in\R^3$ satisfy $|z|\le \frac{1/n}{1+\gamma}$. Put  $z_0 :=\tau(z+y) -\tau(y)$. Then $|z_0|\le|z|$ as $\mathfrak{y}\in A$ and  hence  $z_0\le \frac{1}{n}$,  $\gamma \,|z| \le z_0 +\frac{1}{n}$. Moreover $z_0+y_0=\tau(z+y)$.
Therefore, according to (e), $\mathfrak{x}:= (\tau(z+y),z+y)$ lies in $A_n \cap \Gamma_n^+$, whence $x\in C_{y,n}$. So  $C_{y,n}$ contains  the ball with center $y$ and radius 
$\frac{1/n}{1+\gamma}$. This implies  $\operatorname{diam}(C_{y,n})\ge \frac{2/n}{1+\gamma}>0$ and  (\ref{MTR})\,(iii) for $C=\frac{48}{\pi}\big(1+\frac{1}{\gamma}\big)^3$.\qed
\end{P7T}

 \section{Eigenvalues of $\i\Phi$}\label{EVP}

Here the eigenvalues of $\i\Phi_j(p)$ (\ref{RNWD}) are computed and compared with  the eigenvalues of $\i\Phi^e_j(p)$ (\ref{ECEE}).

\begin{Lem}\label{CEE} Let $p\in\R^3$, $j=1,2,3$, and write $\epsilon$ for $\epsilon(p)$. Put 
\begin{center}
$\lambda_j(p):=\sqrt{\frac{1}{2\epsilon(\epsilon+m)}-\frac{2\epsilon(\epsilon+m)+m^2}{4\epsilon^4(\epsilon+m)^2}p_j^2}$.
\end{center}
Then the eigenvalues of $\i\Phi_j(p)$ are $\pm\lambda_j(p)$.
Both eigenvalues are twofold degenerate. One estimates
\begin{center}
$\frac{m}{2\epsilon^2}\le \lambda_j(p)\le \sqrt{\frac{1}{2\epsilon(\epsilon+m)}} \le \min{ \big\{\frac{1}{2|p|}, \frac{1}{2m}\big\}}$
\end{center}
for all $p\in\R^3$, where $\lambda_j(p)=\frac{m}{2\epsilon^2}$ if $|p_j|=|p|$ and $\lambda_j(0)=\frac{1}{2m}$.
\end{Lem}

{\it Proof.} First one checks that the trace of $\i\Phi_j(p)$ is zero. Hence the sum (with multiplicities) of the eigenvalues is zero. Next one finds
\begin{itemize}
\item[] $(\i\Phi_j(p))^2 =\frac{1}{4\epsilon^2(\epsilon+m)^2}\big( p^2-p_j^2+(\epsilon+m)^2+\frac{1}{\epsilon^2}p_j^2p^2- 2\frac{\epsilon+m}{\epsilon}p_j^2\big)\,I_4$\hfill{(*)}
\end{itemize}
Substitute $p^2$ by $\epsilon^2-m^2$.  The result  on the eigenvalues follows. Regarding the estimations use  $p_j^2\ge 0$ and  $p_j^2\le \epsilon^2-m^2$.
\qed

\begin{Cor} One has
\begin{itemize}
\item $ \lambda_j(p)^2- \lambda_j^e(p)^2 = \frac{1-p_j^2/\epsilon^2}{4\epsilon^2} $
\item $(1-p_j^2/\epsilon^2)\frac{m}{4\epsilon^2} \le  \lambda_j(p)- \lambda_j^e(p) \le (1-p_j^2/\epsilon^2)\frac{1}{2m}$
\end{itemize}
In particular, $\lambda_j^e(p)$ is smaller than  $\lambda_j(p)$ and their difference is smaller than the Compton wavelength. 
\end{Cor}

{\it Proof.} The first item follows from (*) in (\ref{CEE}). The second one follows from the first using the estimations in (\ref{CEE}) and 
$ \frac{m}{2\epsilon^2}\le \lambda_j(p) +\lambda_j^e(p) \le 2\frac{1}{2 m}$.\qed


\begin{thebibliography}{9}


\bibitem{A69} Amrein, W.O.: {\it Localizability for Particles of Mass Zero}, Helv. Phys. Acta \textbf{42}, 149-190 (1969)

\bibitem{A09} Araki, H.: {\it Mathematical Theory of Quantum Fields}, Oxford University Press, (2009)

\bibitem{BK03} Barat, N.,  Kimball J.C.: {\it Localization and Causality for a Free Particle}, Physics Letters A \textbf{308}, 110-115  (2003)

\bibitem{BFM05} Bracken, A.J.,  Flohr, J.A., Melloy, G.F.:  {\it Time-evolution of highly localized positive-energy states of the free Dirac electron}. Proc. R. Soc. A  \textbf{461}, 3633-3645  (2005) 

\bibitem{BM99}  Bracken, A.J., Melloy, G.F.:   {\it Localizing the Relativistic Electron}, J. Phys. A: Math. Gen. \textbf{32},  6127-6139 (1999)

\bibitem{BY94} Buchholz, D., Yngvason, J.: {\it There are no causality problems for Fermi's two atom system}, Phys. Rev. Lett. {\bf 73}, 613-616 (1994)


 \bibitem{BT26} 	B\"{u}rck, I., Tumulka, R.: {\it Dirac Wave Functions of Positive Energy with Arbitrarily Small Position Uncertainty},
 https://doi.org/10.48550/arXiv.2603.04569

 \bibitem{Buschloc1}  Busch, P. and Singh, J. , {\em L\"uders theorem for unsharp quantum measurements.
Physics Letters A 249, 10--12  (1998)

\bibitem{Buschloc2}  Busch, P., {\em Unsharp localization and causality in relativistic quantum theory.} J.
Phys. A: Math. Gen. \textbf{32}, 6535-6546 (1999)}



 \bibitem{BGL97}  Busch, P.,  Grabowski, M., Lahti,  P. J.: {\it Operational Quantum Theory}, Springer-Verlag, Berlin 1997
 


\bibitem{C81} Castrigiano,  D.P.L.:   {\it On Euclidean Systems of Covariance for Massless Particles}, Lett. Math. Phys. \textbf{5}, 303-309 (1981) 



\bibitem{CM82} Castrigiano, D.P.L., Mutze, U.: {\it Covariant description of particle position}, Phys. Rev. D \textbf{26}, 3499-3505 (1982)

 

\bibitem{CL15} Castrigiano, D.P.L.,  Leiseifer, A.D.: {\it Causal Localizations in Relativistic Quantum Mechanics}.  J. Math. Phys. \textbf{56}, 072301 (2015)

\bibitem{C17} Castrigiano, D.P.L.: {\it Dirac and Weyl Fermions - the Only Causal Systems}. arXiv: 1711.06556 (2021)


\bibitem{C23} Castrigiano, D.P.L.: {\it  Causal localizations of the massive scalar boson.} Lett. Math. Phys. \textbf{114}, 2 (2024). https://doi.org/10.1007/s11005-023-01746-z


\bibitem{C24} Castrigiano, D.P.L.: {\it Achronal localization, representations of the causal logic for massive systems.} Letters in Mathematical Physics (2025) 115:25, https://doi.org/10.1007/s11005-025-01911-6


\bibitem{C25} Castrigiano, D.P.L.: {\it Localization of the massive scalar boson on achronal hyperplanes, derivation of Lorentz contraction.} Journal of Mathematical Physics \textbf{66}, 7 (2025). https://doi.org/10.1063/5.0270153



\bibitem{CDM25} Castrigiano, D.P.L., De Rosa, C., Moretti, V.: {\itshape Achronal Localization and Representation of the Causal Logic from a Conserved Current with an Application to Massive Scalar Boson.} Ann. Henri Poincar\'e (2026) https://doi.org/10.1007/s00023-026-01668-1








\bibitem{CJ77} Cegla, W.,  Jadczyk, A.Z.: {\it Causal Logic of Minkowski Space}, Commun. math. Phys. \textbf{57}, 213-217 (1977)








\bibitem{DM24} De Rosa, C., Moretti, V.: {\it Quantum particle localization observables on Cauchy surfaces of Minkowski spacetime and their causal properties.} Letters in Mathematical Physics (2024) 114:72 https://doi.org/10.1007/s11005-024-01817-9

\bibitem{DP00} Dvurečenskij, A.   and  Pulmannová S.:
{\em New Trends in Quantum Structures}, Kluwer 2000.




\bibitem{GGP67} Gerlach, B., Gromes D.,  Petzold J.: {\it  Konstruktion definiter Ausdr\"ucke
f\"ur die Teilchendichte des Klein-Gordon-Feldes}.  Z. Phys. \textbf{204}, 1-11 (1967)

\bibitem{G90} Greiner, W.: {\it Relativistic Quantum Mechanics: Wave Equations}. Springer, Berlin Heidelberg 1990, ISBN 978-3-642-88082-7.



\bibitem{H96} Haag, R.: {\it Local Quantum Physics}, 2nd Ed. Springer-Verlag (1996)



\bibitem{HC01} Halvorson, H., Clifton, B.: {\it No place for particles in relativistic quantum theories?},  arXiv:quant-ph/0103041v1 (2001)

\bibitem{HR80} Hegerfeldt, G.C., Ruijsenaars, N.M.: {\it Remarks on causality, localization, and spreading of wave packets}, Phys. Rev. D \textbf{22}, 377-384 (1980)


\bibitem{H85} Hegerfeldt, G.C.: {\it Violation of Causality in Relativistic Quantum Theory?}, Phys. Rev. Lett. \textbf{54}, 2395-2398 (1985)


\bibitem{H94} Hegerfeldt G.C.: {\it Causality Problems for Fermi's Two-Atom System}, Phys. Rev. Lett. \textbf{72}, 596-599 (1994)

\bibitem{H01}  Hegerfeldt G.C.: {\it Particle Localization and the Notion of Einstein Causality} in {\it Extensions of Quantum Theory}, Eds. A. Horzela, E. Kapuscik, published by Aperion, Montreal 9-16 (2001)














\bibitem{L1910} Lebesgue, H.: {\itshape Sur l'int\'egration des fonctions discontinues}, Annales Scientifiques de l'\'Ecole
Normale Sup\'erieure \textbf{27}, 361-450, (1910) doi:10.24033/asens.624


\bibitem{LO26} Lechner, G., de Oliveira,  I.R.: {\it Causal quantum-mechanical localization observables in lattices of real projections}, 	
https://doi.org/10.48550/arXiv.2602.11392 (2026)

\bibitem{M12} Maggi, F.: {\it Set of finite perimeter and geometric variational problems.} Cambridge University  Press, New York 2012



\bibitem{M02} Melloy, G.F.: {\it The Generalized Representation of Particle Localization in Quantum Mechanics},
Foundations of Physics \textbf{32}, 503-530  (2002)



\bibitem{M23} Moretti, V.: {\it On the relativistic spatial localization for massive real scalar Klein-Gordon quantum particles}. Lett. Math. Phys. \textbf{113}, 66 (2023). https://doi.org/10.1007/s11005-023-01689-5



\bibitem{M26a} Moretti, V.: Spatial Localization of Relativistic Quantum Systems: The Commutativity Requirement and the Locality Principle. Part I: A General Analysis,  	
https://doi.org/10.48550/arXiv.2604.03729 (2026)

 \bibitem{M26b} Moretti, V.: {\em  Spatial Localization of Relativistic Quantum Systems: The Commutativity Requirement and the Locality Principle. Part II: A Model from Local QFT} 	Lett. Math. Phys. (2026) in press\\
https://doi.org/10.48550/arXiv.2604.04173 (2026)









\bibitem{NW49}  Newton,  T.D., Wigner E.P.: {\it Localized States for Elementary Systems}, Rev. Mod. Phys. \textbf{21}, 400-406 (1949)



\bibitem{P77} Pourciau, B.H.: {\itshape   Analysis and Optimization of Lipschitz Continuous Mappings.} JOTA \textbf{22}, 3, 311-351 (1977)


\bibitem{RS79} Reed, M., Simon, B.: {\it Methods of Modern Mathematical Physics}. Academic Press, New York (1979)

\bibitem{S11} Schroer, B.: {\it Causality and Dispersion Relations and the Role of the S-matrix in the ongoing Research}, CBPF-NF-004/11 (2011)





  


\bibitem{T14} Terno, D.R.: {\it Localization of relativistic particles and uncertainty relations}. Phys.
Rev. A \textbf{89}, 042111 (2014)



 \bibitem{T92} Thaller, B.: {\it The Dirac Equation}, Springer-Verlag, Berlin 1992

 \bibitem{WM63} Weidlich, W., Mitra,  A.K.: {\it Some Remarks on the Position Operator in Irreducible Representations of the Lorentz-Group}. Nuovo Cim. \textbf{30}, 385-398 (1963)

\bibitem{W76} Weidmann, J.: {\it Linear Operators in Hilbert spaces}, Springer 1980




\bibitem{W62} Wightman,  A.S.: {\it On the Localizability of Quantum Mechanical Systems}, Rev. Mod. Phys. \textbf{34}, 845-872 (1962)


\bibitem{WS55} Wightman,A.S.,  Schweber, S.S.: {\it Configuration Space Methods in Relativistic Quantum Field Theory I}, Phys. Rev. \textbf{98},  812-837 (1955), and references therein





\end{thebibliography}
\end{document}